\documentclass[11pt,letterpaper,onecolumn,noarxiv]{quantumarticle}
\usepackage[utf8]{inputenc}
\usepackage{amsmath,amsfonts,amssymb,amsthm,bm,mathtools}
\usepackage[backref,colorlinks,bookmarks=true, citecolor=red, urlcolor=black,linkcolor=black, pdfencoding=auto, psdextra]{hyperref}

\usepackage{xcolor}
\usepackage{colortbl}
\usepackage{graphicx}
\usepackage{algorithm,algpseudocode}

\usepackage{ulem}
\usepackage{comment}
\newcommand{\nw}[0]{\lambda}
\newcommand{\dt}[0]{T^{\delta}}
\newcommand{\vJ}[0]{\mathcal{J}}
\newcommand{\vh}[0]{\varrho}
\newcommand{\vbeta}[0]{\eta}
\newcommand{\LN}[0]{{\rm{\ln}}\,}

\newcommand{\cD}[0]{\mathcal{D}}

\newcommand{\cB}[0]{\mathcal{B}}
\newcommand{\cN}[0]{\mathcal{N}}
\newcommand{\cT}[0]{\mathcal{T}}

\newcommand{\cA}[0]{\mathcal{A}}

\newcommand{\cZ}[0]{\mathcal{Z}}

\newcommand{\ket}[1]{\left\vert #1 \right\rangle}
\newcommand{\bra}[1]{\left\langle #1 \right\vert}

\newcommand{\bC}[0]{\mathbb{C}}
\newcommand{\bF}[0]{\mathbb{F}}
\newcommand{\bR}[0]{\mathbb{R}}
\newcommand{\bN}[0]{\mathbb{N}}

\newcommand{\PT}[0]{\mbox{\bf P}}

\newcommand{\Complete}[0]{\mbox{\bf complete}}
\newcommand{\Completeness}[0]{\mbox{\bf completeness}}
\newcommand{\BQP}[0]{\mbox{\bf BQP}}
\newcommand{\BPP}[0]{\mbox{\bf BPP}}

\newcommand{\hard}[0]{\mbox{\bf hard}}
\newcommand{\hardness}[0]{\mbox{\bf hardness}}
\newcommand{\QMA}[0]{\mbox{\bf QMA}}

\newcommand{\StoqMA}[0]{\mbox{\bf StoqMA}}

\newcommand{\PP}[0]{\mbox{\bf PP}}

\newtheorem{definition}{Definition}

\newtheorem{example}[definition]{Example}

\newtheorem{lemma}[definition]{Lemma}
\newtheorem{claim}[definition]{Claim}

\newtheorem{theorem}[definition]{Theorem}
\newtheorem{corollary}[definition]{Corollary}

\title{Positive bias makes tensor-network contraction\newline tractable}
\date{}

\author[1]{Jiaqing Jiang}
\email{jiaqingjiang95@gmail.com}
\affil[1]{CMS, Caltech, Pasadena,  CA, 91125, USA}

\author[2]{Jielun Chen}
\email{jchen9@caltech.edu}
\affil[2]{Department of Physics, Caltech, Pasadena, CA 91125, USA}

\author[3,4]{Norbert Schuch}
\email{norbert.schuch@univie.ac.at}
\affil[3]{University of Vienna, Faculty of Mathematics, Oskar-Morgenstern-Platz 1, 1090 Wien, Austria}
\affil[4]{University of Vienna, Faculty of Physics, Boltzmanngasse 5, 1090 Wien, Austria}

\author[5,6]{Dominik Hangleiter}
\email{mail@dhangleiter.eu}
\affil[5]{QuICS, University of Maryland \& NIST, College Park, MD 20742, USA}
\affil[6]{Simons Institute for the Theory of Computing, University of California, Berkeley, CA 94720, USA}

\begin{document}

\maketitle

%\tableofcontents

\newcommand{\ns}[1]{\textcolor{red}{\textbf{[NS: #1]}}}

\begin{abstract}

Tensor network contraction is a powerful computational tool in quantum 
many-body physics, quantum information and quantum chemistry. The complexity of contracting a tensor network is thought to mainly depend on its entanglement properties, as reflected by the Schmidt rank across bipartite cuts.
Here, we study how the complexity of tensor-network contraction depends on a different notion of quantumness, namely, the sign structure of its entries.
We tackle this question rigorously  by investigating the complexity of contracting tensor networks whose   entries have a positive bias. 

We show that for intermediate bond dimension $d\gtrsim n$,  a small positive mean value $ \gtrsim 1/d$ of the tensor entries already dramatically decreases the computational complexity of \emph{approximately} contracting random tensor networks, enabling a quasi-polynomial time algorithm for arbitrary $1/\mathrm{poly}(n)$ multiplicative approximation.
At the same time \emph{exactly} contracting such tensor networks remains $\#\PT$-$\hard$, like for the zero-mean case~\cite{haferkamp2020contracting}.
The mean value $1/d$ matches the phase transition point observed  in~\cite{chen2024sign}. 
Our proof makes use of Barvinok's method for approximate counting and the technique of mapping random instances to statistical mechanical models. 
We further consider the worst-case complexity of approximate contraction of positive tensor networks, where all entries are non-negative.  
We first give a simple proof showing that  a multiplicative approximation  with error  exponentially close to one  is at least $\StoqMA$-$\hard$.
We then show that when considering  additive error in the matrix $1$-norm, the contraction of positive tensor network is $\BPP$-$\Complete$. 
This result compares to Arad and Landau's~\cite{arad2010quantum} result, which shows that for general tensor networks, approximate contraction up to  matrix $2$-norm additive error is $\BQP$-$\Complete$. 

Our work thus identifies new parameter regimes in terms of the positivity of the tensor entries in which tensor networks can be (nearly) efficiently contracted.
\end{abstract}

\newpage
\tableofcontents

\section{Introduction}

Tensor network contraction is a powerful computational tool for studying quantum information and quantum many-body systems. It is widely used in  estimating ground state properties~\cite{white1993density,white1992density,murg2007variational,vlaar2021simulation}, approximating partition functions~\cite{evenbly2015tensor,zhao2010renormalization}, simulating evolution of 
quantum circuits~\cite{markov2008simulating,pednault2017breaking,huang2020classical},   
as well as decoding for  quantum error correcting codes~\cite{ferris2014tensor,bravyi2014efficient}. 
Mathematically, a tensor network $T\coloneqq T(G,M)$ on a graph $G=(V,E)$ can be interpreted as an edge labeling model.
 Each edge can be labeled by one of $d$ different colors, where~$d$ is  called the \textit{bond dimension}.
 Each vertex~$v$  is associated with a function $M^{[v]}$, called \textit{tensor}, 
 whose value depends on the labels of edges adjacent to~$v$. The tensor $M^{[v]}$ can be represented as a vector by enumerating its values with respect to various edge labeling.
For any edge labeling $c$, denote the 
value (entry) of the tensor $M^{[v]}$ by  $M^{[v]}_c$. 
The \textit{contraction value of tensor network} 
 is defined to be
\begin{align}
	\chi(T)\coloneqq \sum_{\text{edge labeling $c$}\,\,\,}	 \prod_{v \in V} M^{[v]}_c\label{eq:TNDef}.
\end{align}
In applications of tensor networks, the contraction value represents the quantities of interest and the goal of   tensor-network contraction algorithms is to compute the contraction value to high precision.

It is therefore a fundamental question to determine when $\chi(T)$ can be computed efficiently.
Despite the practical and foundational importance of this question, unfortunately most  rigorous results show that tensor network contraction is extremely hard, with very few tractable cases known, that is, cases for which a 
(quasi-)polynomial time algorithm exists.  Specifically, it is well-known that
 computing $\chi(T)$ exactly is $\#\PT$-$\hard$~\cite{schuch2007computational} and therefore intractable in the worst case.
 The hardness can be further strengthened to the average case, where Haferkamp \textit{et al.}~\cite{haferkamp2020contracting}  showed that  even  for random tensor networks  on a 2D lattice, 
 computing $\chi(T)$ exactly remains $\#\PT$-$\hard$ for typical instances. 
 There, the randomness is modeled by sampling  the entries of the tensor network iid.\ from a  Gaussian distribution with \textit{zero} mean and unit variance. Conversely, (quasi-)polynomial time algorithms
 are only known for restricted cases, like tensor networks on simple graphs of small tree-width~\cite{markov2008simulating}, for example 1D line or tree; or for 
 restrictive symmetric tensor network~\cite{patel2017deterministic} where 
each entry is very close to $1$, which requires that $\forall c, v, |M^{[v]}_c -1|\leq 0.35/(\Delta+1)$, where $\Delta$ is the maximum degree of the graph. % Theorem 1.4 of the paper. 
Besides, for tensor networks with  uniformly gapped parent Hamiltonians,   (quasi)-polynomial time algorithm is known for computing  local
expectation values~\cite{schwarz2017approximating}.   

But while efficient and provably correct tensor-network  contraction algorithms are rare, for many many-body physics applications, 
state-of-the-art numerical algorithms achieve desired accuracy in practice~\cite{orus2019tensor,banuls2023}.
To obtain a better understanding of when and why such heuristics  work, it is important to identify new tractable cases in tensor network contraction.  
With this goal in mind, a recent line of  
work suggests an interesting direction, namely, that the sign structure of the tensor entries influences the entanglement and therefore affects the the complexity of tensor network contraction~\cite{gray2022hyper,chen2024sign}. 
In particular, it has been observed  that there is a sharp phase transition in the entanglement thus the complexity of approximating random tensor networks, when the mean of the entries is shifted from zero to positive~\cite{gray2022hyper,chen2024sign}.

\subsection{Main results and technical highlights} 
In this work, we rigorously investigate the impact of  sign structure on the complexity of tensor network contraction in various regimes. We mainly focus on the contraction of the physically motivated 2D tensor networks, which 
 are widely used as ground state ansatzes for local Hamiltonians~\cite{corboz2016variational,vanderstraeten2016gradient} (Projected Entangled Pair States) and for the simulation of quantum circuits~\cite{guo2019general}.
 
Recall that for random 2D tensor network whose entry has zero mean, the exact contraction is $\#\PT$-$\hard$~\cite{haferkamp2020contracting}. We first show that a positive bias does not decrease the complexity of the exact contraction:
\begin{theorem}[Informal version of Theorem \ref{thm:exactH}]\label{intro:exact}
    The  exact contraction of random 2D tensor network 
    whose entries are   iid.\ sampled from a  Gaussian distribution with \underline{positive} mean and unit variance remains $\#\PT$-$\hard$.
\end{theorem}
While Theorem \ref{intro:exact} indicates 
the exact contraction remains hard, our  main result is proving that  a small positive mean   significantly decreases the computational complexity of multiplicative approximation, enabling a quasi-polynomial time algorithm. This provides   rigorous evidence that the sign structure of the tensor entries influences the contraction complexity, as  observed and conjectured in previous  works~\cite{gray2022hyper,chen2024sign}.  
In particular,  we show that
\begin{theorem}[Informal version of Theorem \ref{thm:random}]\label{intro:random}
    For random 2D tensor network with intermediate bond dimension $d \gtrsim n$, where the entries are iid.\  sampled from Gaussian distribution with mean 
    $\mu \gtrsim 1/d$ and unit variance, there exists  a quasi-polynomial time algorithm which with high probability  approximates the contraction value up to arbitrary $1/poly$ multiplicative error.  
\end{theorem}
Here $a \gtrsim b$ means that $a$ scales at least as fast as $b$.

While it is expected that  tensor network contraction becomes easier  when all entries are positive so that there is  no sign problem, our result  is much more fine-grained than this belief since our tensor network is only \textit{slightly} positive, that is, a significant portion of the tensor entries are still negative.
In particular, note that
the mean value $\gtrsim 1/d$ is far less than the unit variance of the tensor entries. 
Compared to previous work~\cite{patel2017deterministic} which shows that tensor-networks whose all entries are close to $1$ can be contracted using Barvinok's method,\footnote{More precisely for 2D tensor network, it requires that $\forall c, v, |M^{[v]}_c-1|\leq 0.35/(4+1)=0.07$. } our result allows the entries to have significant fluctuations and to be a mixture of positive and negative values.  
We also note  that the threshold value $\gtrsim 1/d$ matches the phase transition point predicted in~\cite{chen2024sign}\footnote{To clarify, \cite{chen2024sign} draws each tensor from a Haar random distribution. If one does the same calculation for drawing each entry from Gaussian random distribution, the
predicted phase transition point will also be approximately $1/d$. }
  with respect to the  entanglement-based
contraction algorithm. The fact that two different methods (our algorithm and the entanglement-based algorithm) admit the same threshold might indicate that there is a genuine phase transition in the complexity of tensor network contraction  at this point. 
The requirement of $d \gtrsim n$ on the bond dimension in Theorem \ref{intro:random} is due to the fact that certain concentration effects set in at $ d \sim n$. 
One may wonder then whether the intermediate bond dimension and the nonzero mean make the mean contraction value  $\mu^n d^{2n}$ (attained when all entries in the tensor network take the mean value $\mu$) a precise guess for the contraction  value, that is $\chi(T)=\mu^nd^{2n}(1+1/poly(n))$.
This is not the case since a simple lower bound shows that the second moment of  $\chi(T)/(\mu^nd^{2n})$ is at least $2$. In comparison, our algorithm can achieve an arbitrary $1/poly(n)$ multiplicative error in quasi-polynomial runtime; recall Theorem~\ref{intro:random}.   
Besides, although Theorem \ref{intro:random} is formulated for random 2D tensor networks, the proposed algorithm is well-defined and runs in quasi-polynomial time for an arbitrary graph $G$ of constant degree, which
may inspire new heuristic algorithms for general tensor networks.

Besides studying the average case complexity for  approximating slightly positive tensor networks, we also investigate the complexity of approximating (fully) positive tensor networks, where all the entries are positive.  
 Approximate contraction of positive tensor network is directly related to approximate counting,  we  give a simple proof to show that    
\begin{theorem}[Informal version of Theorem \ref{thm:StoqMA}]
   $1/poly(n)$ multiplicative 
 approximation of positive tensor network is $\StoqMA$-$\hard$. The $\StoqMA$-$\hard$ remains even if we relax the multiplicative error from $1/poly(n)$ to a value exponentially close to one.
\end{theorem}
Here, $\StoqMA$ is the complexity class whose canonical complete problem is to decide the ground energy for stoquastic Hamiltonians~\cite{bravyi2006merlin}.

In addition to multiplicative  approximation, 
 we also investigate
the impact of sign in the hardness of tensor network contraction w.r.t.\ \textit{certain additive} error. In particular, previously  Arad and Landau~\cite{arad2010quantum} showed that approximating the contraction value w.r.t.\ the matrix $2$-norm additive error is 
equivalent to quantum computation, that is $\BQP$-$\Complete$.  
In contrast, we prove that if the tensor network is \textit{positive}, where all entries are non-negative, then  approximating the contraction value w.r.t matrix $1$-norm additive error is 
equivalent to classical computation, that is $\BPP$-$\Complete$. 

\begin{theorem}[Informal version of Theorem~\ref{thm:BPP}]\label{intro:BPP}
Given a positive tensor network $T\coloneqq T(G,M)$ on a constant-degree graph $G$.   Given an arbitrary order of the vertex $\{v\}_v$, one can view each tensor $M^{[v]}$ as a matrix $O^{[v]}$ by specifying the in-edges and out-edges. 
 It is $\BPP$-$\Complete$ to 
 estimate $\chi(T)$ with additive error $\epsilon \Delta_1 $, for $\Delta_1\coloneqq  \prod_v \|O^{[v]}\|_1$ and $\epsilon=1/poly(n).$
	\end{theorem}
Technically
 \cite{arad2010quantum} simulates general matrix multiplication by quantum circuits. In  Theorem \ref{intro:BPP} we simulate non-negative matrix multiplication by random walks.

\textbf{Technical highlights.}
Our main technical contribution is to show that a small mean value dramatically decreases the complexity of approximate contraction. 
Our result
significantly extends the regime in which efficient
approximate contraction algorithms for  tensor networks are known.
This is formalized and proved in  Theorem \ref{intro:random}. 
The algorithm in Theorem \ref{intro:random} differs from  commonly used numerical algorithms for   tensor network contraction, which are based on the truncation of  singular value decomposition and whose performance is determined by
entanglement properties~\cite{hastings2007area,arad2017rigorous}. 
Instead, for Theorem~\ref{intro:random} we use Barvinok's method from approximate counting. 
This method has previously been used for approximating the permanent, the hafnian~\cite{barvinok2016approximating,barvinok2016computing} and partition functions~\cite{barvinok2014computing,patel2017deterministic}.

At a high level, Barvinok's method interprets the contraction value $\chi(T)$ as a polynomial $G(z)$  where $G(1)=\chi(T)$, and uses  Taylor expansion of $\ln G(z)$ at $z=0$ to get an additive error approximation of $\ln G(1)$, thus  an
multiplicative approximation of $\chi(T)$. The key technical part of applying Barvinok's method to different tasks is proving the corresponding $G(z)$ is  root-free in the disk centered at $0$ with radius slightly larger than $1$, which ensures that $\ln G(z)$ is analytic in this disk. Denote this disk as $\cB$.
Previously Patel and Regts~\cite{patel2017deterministic} had applied Barvinok's method to symmetric tensor networks where all the entries are close to $1$ within error $0.35/5=0.07$, by proving that $G(z)$ is root-free in  $\cB$. Our setting allows entries to  have significant fluctuations, thus the root-free proof in \cite{patel2017deterministic} does not apply. We circumvent this problem using the following two ideas:
\begin{itemize}
	\item \textbf{Root-free strip inspired from  approximating random permanent}. 
	Instead of applying Barvinok's method directly and  proving $G(z)$ is  root-free in the disk $\cB$,  we apply a variant of Barvinok's method used for approximating random permanents by Eldar and Mehraban~\cite{eldar2018approximating}. 
	There, the idea is to use Jensen's formula to find a root-free \textit{strip} connecting $0$ and $1$. The advantage of this variant is that it allows for a constant number of zeros in the unit disk $\cB$ as long as there is a root-free path of some width connecting $0$ and $1$. We notice that this method from approximating permanent can also be applied to random tensor networks. In particular, using Jensen’s formula~\cite{eldar2018approximating},
 the number of roots in $\cB$ can be bounded by estimating the second moment $E |h(z)|^2$, where $h(z)$ is
 a rescaled version of $G(z)$, and $E$ denotes the expectation value over the randomness of the tensor network.	
 Besides,
compared to \cite{eldar2018approximating}, in our setting we use a different and much simpler method to find the root-free strip.
	\item \textbf{Mapping random instance to statistical mechanical model.} Since we are working on random tensor networks, the technique used by Eldar and Mehraban~\cite{eldar2018approximating} to bound $E |h(z)|^2$ for random permanents fails entirely. 
	To bound $E |h(z)|^2$  for random tensor networks, we adapt a technique of mapping random instances to a classical statistical mechanical model (statmech model). 
	This technique has been used in the physics literature to study phase transitions in random tensor networks~\cite{yang2022entanglement,levy2021entanglement,hayden2016holographic} and random circuits~\cite{bao2020theory,bao2021finite}. 

Although
in general such mapping and the properties of the statmech model like its partition function are hard to analyze, and heuristic approximations are needed in many related literature, we notice that in our application the statmech model is simple enough to obtain a rigorous result.  
In particular, we show that $E |h(z)|^2$ is proportional to the partition function of a 2D Ising model with magnetic field parameterized by $z$. Then we further use the
 finite-size variant of the Onsager solution of the 2D Ising model~\cite{kaufman1949crystal,majumdar1966analytic} 
to get a decent estimate of $E |h(z)|^2$ for relevant ranges of $z$, allowing the Barvinok method to be applied.

\end{itemize}

\subsection{Conclusions and open problems}\label{sec:open}

We  investigate how the contraction complexity of tensor networks depends on the sign structure of the tensor entries. 
For random tensor networks in 2D, we show that there is a quasi-polynomial time approximation algorithm if the entries are drawn with a small nonzero mean and
intermediate bond dimension. 
At the same time, exactly computing the contraction value in this setting remains $\#\PT$-$\hard$. Our work thus provides rigorous evidence for 
the  observations~\cite{gray2022hyper,chen2024sign} that shifting the mean by a small amount away from zero dramatically decreases the contraction complexity. 
Compared to \cite{patel2017deterministic} which similarly uses Barvinok method but 
requires all entries to be close to $1$, our setting  allows significant fluctuations in the entries and  greatly extends the known region where (quasi-)polynomial time average-case contraction algorithms exist. 
While it is expected that tensor network contraction becomes easier when all entries are positive, 
our result suggests that even for \textit{slightly} positive tensor networks, one can still utilize the sign structure to obtain a (quasi-)efficient algorithm. 
Moreover,  
\cite{chen2024sign} 
observed  that the standard entanglement-based contraction algorithm starts working at $\mu \gtrsim 1/d$. 
We show that a \textit{completely different} rigorous Barvinok-based algorithm also starts working at $\mu \gtrsim 1/d$. 
This might indicate that there is a genuine phase transition in the complexity of tensor-network contraction happening here.

Indeed, we also assess the worst-case complexity of approximating fully positive tensor networks. 
Specifically, we prove that approximating the contraction value  of positive tensor networks multiplicative error close to unity is \StoqMA-hard.
But when requiring only an inverse polynomial additive error in matrix $1$-norm there exists an efficient classical algorithm.

Our work initiates the rigorous study of how the computational difficulty of contracting tensor networks depends on the sign structure of the tensor entries. 
If one views the hardness of contraction as a function of mean value and bond dimension, while we identify a new tractable region, there are many open questions left.
\begin{itemize}
    \item First, while our approximation algorithm based on Barvinok's method works for typical instances, it remains an open question to what extent a positive bias can ease practical  tensor network contraction.
It would therefore be interesting to understand whether our algorithm or variations of it can aid in practically interesting cases. 
\item Moreover, our current proof works for intermediate bond dimension but  not  for constant bond dimension.
Potentially, techniques like cluster expansion~\cite{mann2021efficient,helmuth2019algorithmic}  may be used  to design new contraction algorithm for constant bond dimension,  proving a correspondent of  Theorem \ref{intro:BPP} for that setting. 
It might be worth mentioning that a direct application of cluster expansion does not work, where one can prove the expansion series is not absolutely convergent. 
More refined techniques are thus required. 
\item 
Finally, although  current numerical algorithms have poor performance for zero-mean tensor network contraction, there is no known rigorous complexity result to establish the hardness of approximate contraction.
\end{itemize}

In the context of approximating fully positive tensor networks, it would be interesting to see whether there exists an efficient classical algorithm that can achieve the same ($2$-norm) precision as a quantum computer for positive tensor networks, or whether there is a room for quantum advantage even for positive tensor networks.\footnote{We acknowledge Zeph Landau for raising this question.}

\subsection{Structure of the manuscript}

The structure of this manuscript is as follows. In Section \ref{sec:Notation} we define notations and  tensor networks. In Section \ref{sec:barvinok} we review Barvinok's method and its variant.  In section 
\ref{sec:alg} we adapt Barvinok's method to tensor network contraction. 
In Section \ref{sec:TN} we give a quasi-polynomial time algorithm for approximating random 2D tensor networks with small mean and intermediate bond dimension. In Section \ref{sec:positive} we prove the results concerning approximating positive tensor networks.

\section{Notation and tensor networks}\label{sec:Notation}
In this section, we introduce necessary notations and definitions for tensor networks.

\noindent\textbf{Notation.}
We use $[k]$ for $\{0,1,\ldots ,k-1\}$. 
We use $\overline{z}$ to denote its complex conjugate.  For  $v\in \bC$ and $\epsilon\in \bR$, we say $\hat{v}$ approximates $v$ with $\epsilon$-multiplicative error if $|v-\hat{v}|\leq \epsilon |v|$.   For $x\in\{\pm 1\}^n$, we use $|x|$ to denote the number of $-1$ in $x$.  We use $\delta_{ij}$ for the delta function, where $\delta_{ij}=1$ if $i=j$ and equals $0$ otherwise.

For a matrix $A\in \bC^{s\times t}$, the matrix $p$-norm is defined as 
\begin{align}
    \|A\|_p \coloneqq \sup_{x\neq 0, x\in \bR^t}\frac{\|Ax\|_p}{\|x\|_p}
\end{align}
The $2$-norm $\|A\|_2$ is known as the spectral norm. The  $1$-norm equals to the maximum of the absolute column sum, that is
    $$\|A\|_1 = \max_{1\leq j\leq t} \sum_{i=1}^s |A_{ij}|.$$

For $\mu,\sigma\in \bR$,
we use  $X\sim \cN_{\bR}(\mu,\sigma^2)$ to denote that the random variable $X$ is sampled from  the Gaussian distribution with mean $\mu$ and standard derivation $\sigma$.
For $\mu\in \bC$, we use $\Re(\mu),\Im(\mu)\in \bR$ to denote the real and imaginary part of $\mu$, i.e. $\mu= \Re(\mu) +\Im(\mu) i$. We use $X\sim \cN_{\bC}\left(\mu,\sigma^2\right)$ if 
$$\Re\left(X \right)\sim \cN_{\bR}\left(\Re(\mu),\frac{\sigma}{2}^2\right), \Im(X)\sim \cN_{\bR}\left(\Im(\mu),\frac{\sigma}{2}^2\right).$$

\vspace{1em}

\noindent\textbf{Tensors and tensor networks.}     A \textit{tensor} $N$ of \textit{rank} $k$ and bond dimension $d$ is an array of $d^k$ complex numbers which is indexed by $N_{i_1,i_2,\ldots ,i_k}$, where $i_s$ takes values from $[d]$ for $1\leq s \leq k$. We call the complex numbers in the array the \textit{entries} of the tensor $N$. We use $\overline{N}$ to denote the tensor obtained by complex conjugating every entry of $N$. 
For two tensors $N$ and $M$ with the same rank and bond dimension, the addition $N + M$ is a new tensor obtained by addition of the two arrays. 
For convenience, in the rest of the paper we assume that all the indices have the same dimension $d$.\footnote{This is not a restriction, since we can just take $d$ to be the maximum  dimension of all indices in the tensor network, and introduce dummy dimensions elsewhere.} We will always assume $d=O(poly(n))$.

A \textit{tensor network} $T\coloneqq T(G,M)$  is described by an $n$-vertex graph $G=(V,E)$ and a set of tensors on vertices, denoted as $M=\{M^{[v]}\}_v$.  More specifically,
on each vertex $v$ of degree $k_v$ there is a tensor $M^{[v]}$ of rank $k_v$, where the indices $i_1,\ldots, i_{k_v}$ correspond to $k_v$ edges.   
One can interpret $[d]$ as $d$ different colors, and $i_s\in[d]$ represents that we label the corresponding edge with color $i_s$.  
 Denote this  edge labeling as $c: E\rightarrow [d]^{|E|}$,
 we write 
$M^{[v]}_c  \coloneqq  M^{[v]}_{i_1,\ldots,i_{k_v}}.$   
With an arbitrary ordering of edges, we can conceive of the labeling $c$ as a vector $c\in [d]^{|E|}$.
The \textit{contraction value of tensor network}  is then defined to be
\begin{align}\label{eq:TN}
    \chi(T)  \coloneqq \sum_{c\in [d]^{|E|} }  \prod_{v\in V} M^{[v]}_c.
\end{align}

\noindent\textbf{Product and contraction.} Besides Eq.~(\ref{eq:TN}),
another equivalent way of defining the contraction value of tensor network $\chi(T)$ is via a graphical representation, which is more intuitive and will be used in the proofs. 
 As in Figure \ref{fig:tensor3} (a), for a tensor of rank $k$,  we represent it as a vertex with $k$ edges. We term such edges which  connect to only one vertex \textit{free edges.}
\begin{figure}
    \centering
\includegraphics[width=0.6\textwidth]{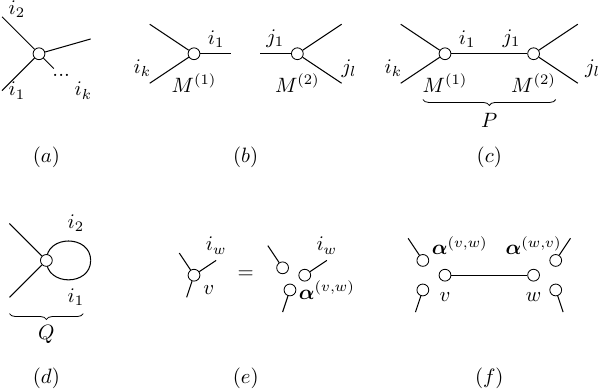}
\caption{Tensor and operations on tensors.  (a) A rank-$k$ tensor. (b) The product of a rank-$k$ and a rank-$l$ tensor. (c) Contracting two tensors by identifying edges $i_1$ and $j_1$. 
(d) Contracting two free edges in the same tensor. (e) A special tensor which can be  factorized into a product of rank-$1$ tensors. (f) If all tensors have a factorized structure, then the contraction value of tensor network can be computed by contracting the rank-$1$ tensors.}
\label{fig:tensor3}
\end{figure}

With this graphical representation, we introduce  two operations on tensors.  
Consider a  tensor $M^{(1)}$ of rank $k$, with free edges indexed by  $i_1,\ldots ,i_k$, and another tensor $M^{(2)}$ of rank $l$, with free edges indexed by $j_1,\ldots ,j_l$.  We use Figure
\ref{fig:tensor3} (b) to represent the \textit{product}  of $M^{(1)},M^{(2)}$, that is  a new tensor $ M^{(1)}\otimes M^{(2)}$  of rank $k+l$  and with free edges indexed by $i_1,\ldots ,i_k;j_1,\ldots ,j_l$,  
where 
\begin{align}
   \left( M^{(1)}\otimes M^{(2)}\right)_{i_1,\ldots ,i_k, j_1,\ldots ,j_l} \coloneqq  M^{(1)}_{i_1,\ldots ,i_k}  M^{(2)}_{ j_1,\ldots ,j_l}.
\end{align} 
The product operation can be generalized to multiple tensors recursively, 
 \begin{align}
     M^{(1)}\otimes M^{(2)}\otimes M^{(3)}\otimes M^{(4)} \ldots  \coloneqq   \left( \left( \left( M^{(1)}\otimes M^{(2)} \right) \otimes M^{(3)}\right)\otimes M^{(4)}\right) \ldots 
 \end{align}
 One can check that the order of this recursion does not change the final tensor.
 
Another operation which defines a new tensor is \textit{contraction}, that is, connecting different tensors by identifying a free edge of one tensor with a free edge of another tensor and summing over that index.
 Starting from the two tensors $M^{(1)}$ and $M^{(2)}$, contracting the indices $i_1$ and $j_1$ results in a new tensor $P$ of rank $k+l-2$,  with free edges indexed by $i_2,\ldots ,i_k,j_2,\ldots ,j_l$, where
\begin{align}
    P_{i_2,\ldots ,i_k,j_2,\ldots ,j_l} = \sum_{f\in [d]} M^{(1)}_{f,i_2,\ldots ,i_k} M^{(2)}_{f,j_2,\ldots ,j_l}.
\end{align}
Graphically, this operation is represented by joining the two contracted edges, see Figure \ref{fig:tensor3} (c).

One can also contract two free edges in the  same tensor. 
Consider the contraction of the indices $i_1,i_2$ of $M^{(1)}$.
Figure \ref{fig:tensor3} (d) represents a new tensor $Q$ of rank $k-2$ and with free edges $i_3,\ldots ,i_k$ where 
\begin{align}
    Q_{i_3,\ldots ,i_k} =\sum_{s\in [d]} M^{(1)}_{s,s,i_3,\ldots ,i_k}.\label{eq:trace}
\end{align}
The contraction operations can be generalized to contracting multiple pairs of edges by contracting the pairs one by one. 
Note that the order of contraction does not change the final tensor.

One can check that given a tensor network $T=T(G,M)$, the contraction value of tensor network defined by Eq.~\eqref{eq:TN} is equal to the value obtained when contracting $\bigotimes_v M^{[v]}$ by identifying the free edges according to the edges of $G$.

 For any vertex $v$, use $N(v)$ for the vertices  adjacent to $v$ in $G$.
\begin{example}
\label{ex:allone product}
Here we give an example of how the graphical representation simplifies the computation of the contraction value. Consider a case in which each $M^{[v]}$ has  a factorized structure, that is,  there exist vectors $\bm{\alpha}^{(v,w)}\in \bC^d$ for $w\in N(v)$ such that 
 $$M^{[v]}=\bigotimes_{w\in N(v)} \bm{\alpha}^{(v,w)},$$ equivalently the entry $$M^{[v]}_{\ldots ,i_w,\ldots }= \prod_{w\in N(v)}\bm{\alpha}^{(v,w)}_{i_w}.$$ Then $M^{[v]}$ can be represented by a product of $|N(v)|$ tensors as shown in Figure \ref{fig:tensor3} (e). As a consequence, one can check that in this special case computing $\chi(T)$ is easy:
 as in Figure \ref{fig:tensor3} (f), 
  one can write $\chi(T)$ in a factorized way, where each edge $(v,w)$ contributes a factor $\langle \bm{\alpha}^{(v,w)},\bm{\alpha}^{(w,v)}\rangle\coloneqq \sum_{f\in [d]} \bm{\alpha}^{(v,w)}_f\bm{\alpha}^{(w,v)}_f$ as follows:
 \begin{align}
	\chi(T) &= \sum_{\text{all edge labeling $c$}}	 \prod_v M^{[v]}_c\\
	&= \prod_{\text{edges $(v,w)$}}\langle \bm{\alpha}^{(v,w)},\bm{\alpha}^{(w,v)}\rangle.\label{eq:rank1}
\end{align}
\end{example}
\vspace{1em}

\noindent\textbf{2D tensor network.} We call a tensor network  $T=T(G,M)$ a \textit{2D tensor network} if 
 the graph $G$ is a 2D  lattice. 
We assume the lattice has  size $L_1\times L_2$ with $n=L_1\times L_2$,  and satisfies periodic boundary conditions, that is  can be mapped onto a torus.   The periodic boundary condition is mainly to ease the analysis. 
In particular, every vertex has degree $4$. For simplicity, we assume that $L_2$ is even.

\begin{itemize}
    \item 
 For $\mu>0$, we define a \textit{2D $(\mu,n,d)$-Gaussian tensor network} $T(G,M)$ as
   an $n$-vertex 2D tensor network with  bond dimension $d$, where the entries of every tensor $M^{[v]}$ are iid.\ sampled from the complex Gaussian distribution $\cN_\bC(\mu,1)$, i.e.\ 
   \begin{align}
       (M^{[v]})_{i_1,i_2,i_3,i_4}  \stackrel{i.i.d.}{\sim}\cN_{\bC}(\mu,1).
   \end{align}
 \item For technical reasons, for  $z\in \bC$ we also  define the
\textit{2D $(z,n,d)$-shifted-Gaussian tensor network} $T(G,M)$, which   is an $n$-vertex 2D tensor network with bond dimension $d$: For every vertex $v$,  Let $(A^{[v]})_{i_1,i_2,i_3,i_4}\stackrel{i.i.d.}{\sim}\cN_{\bC}(0,1)$,  
the entries of $M^{[v]}$ are defined to be
\begin{align}
    &(M^{[v]})_{i_1,i_2,i_3,i_4}\coloneqq  1+z \cdot (A^{[v]})_{i_1,i_2,i_3,i_4}.
\end{align}
We write the tensor $M^{[v]}$ as 
\begin{align}
\label{eq:M-is-J-plus-zA}
M^{[v]} = J^{[v]} + z A^{[v]}	
\end{align}
where $J^{[v]}$ is a tensor whose entries are all $1$. Note that $J^{[v]}$ has a factorized structure 
$$J^{[v]}= [1,\ldots ,1]^{\otimes 4}.$$
 We abbreviate the  2D $(z,n,d)$-shifted-Gaussian  tensor network $T(G,M)$ as $T_A(z)$ where $A\coloneqq \{A^{[v]}\}_v$.
\end{itemize}

\section{Barvinok's method and its variant}\label{sec:barvinok}

In this section we  review Barvinok's method, which was first developed by Barvinok~\cite{barvinok2016approximating,barvinok2016computing}, and is a general method for approximate counting. 
It has been applied to  approximating  permanents~\cite{eldar2018approximating}, hafnians~\cite{barvinok2016approximating,barvinok2016computing} and partition functions~\cite{barvinok2014computing,patel2017deterministic}. In particular, Barvinok's method was applied to contracting symmetric tensor networks where all entries are very close to $1$~\cite{patel2017deterministic}. Our setting allows the entries  having significant fluctuations where the standard Barvinok's method fails. Instead our algorithm  builds  from a special variant of Barvinok's method  used in approximating random permanents~\cite{eldar2018approximating}, which we summarize below. All Lemmas and Theorems quoted here are proven in \cite{barvinok2016approximating,barvinok2016computing,eldar2018approximating}.

Roughly speaking, the idea of Barvinok's method is to approximate an analytic  function via its Taylor series around $0$. The performance of this approximation depends on the location of the roots of the analytic  function.

 Consider a polynomial  $G(z)$ of degree $n$, where $G(z)\neq 0$ for $z$ on a simply connected open area containing $0$ in the complex plain. We choose the branch of the complex logarithm, denoted as $\LN$, such that $\LN G(0)$ is real. Define $F(z):=\LN G(z)$. In our application, $G(1)$ will encode the  contraction value of tensor network. An additive approximation of $F(1)$ will give a multiplicative approximation to $G(1)$.
 For $r,w>0$,
we  use $\cB(r) \subset \mathbb C$ to denote the 
the disk of radius~$r$ centered at~$0$,  and  use $\cT(re^{i\theta},w)$ to denote the strip  of width $w$ around the line between $0$ and $re^{i\theta}$, that is
\begin{align}
    &\cB(r)\coloneqq \{z\in \bC \,\,\big|\,\, |z|\leq r\},\nonumber\\
    &\cT(re^{i\theta},w)\coloneqq \{z\in \bC \,\,\big|\,\,  -w \leq \Re(ze^{-i\theta})\leq r+w, \quad |\Im(z e^{-i\theta})|\leq w 
    \}.   \nonumber 
\end{align}
 The following lemma quantifies the approximation error incurred by approximating $F(z)$  using a root-free disk of $G(z)$.
% Another proof can be seen at  \href{https://homes.cs.washington.edu/~shayan/courses/sampling/counting-14.pdf}{lec notes}
 \begin{lemma}[Approximation using a root-free disk, see the proof of Lemma 1.2 in \cite{barvinok2016computing}]\label{lem:rootfree} Let $G(z)$ be a polynomial of degree $n$ and suppose $G(z)\neq 0$ for all $|z|\leq \vbeta$ where $\vbeta>1$. Let $F(z)\coloneqq  \LN G(z).$ Then $F(z)$ is analytic for $|z|\leq 1$. Moreover, consider a degree $m$ Taylor approximation of $F(z)$, 
 \begin{align}
 	P_m(z) \coloneqq  F(0) +\sum_{k=1}^m \frac{\partial^k F(z)}{\partial z^k}  \bigg|_{z=0} \frac{z^k}{k!}
 \end{align}
 Then, for all $|z|\leq 1$,
 \begin{align}
 	|F(z)-P_m(z)| 
 & \leq 	\frac{n}{(m+1)\vbeta^m(\vbeta-1)}.
 \end{align}
 \end{lemma}

Recall that additive approximation of $F(z)$ implies  multiplicative approximation of $G(z)$.
  To translate Lemma \ref{lem:rootfree} into 
an efficient algorithm, one further needs to efficiently compute the first few derivatives of $F$. 
Barvinok shows that the derivatives of $F$ can be efficiently computed using the derivatives of $G$.
 \begin{lemma}[\cite{barvinok2016approximating}]\label{lem:fg} Let $G(z)$ be a polynomial of degree $n$  and $G(z)\neq 0$. If one can compute the first $l$ derivatives of $G(z)$ at $z=0$  in time $t(n)$,  then one can compute the first $l$ derivatives of $F(z) \coloneqq \LN  G(z)$ at $z=0$ in time $O(l^2 t(n))$.
 \end{lemma}

Lemma \ref{lem:rootfree} implies that the Taylor series at $z=0$ gives a good approximation to $F(z)=\LN G(z)$, as long as $G(z)$ is root-free in a
  disk centered at  $0$ that contains~$z$.
Lemma~\ref{lem:rootfree} can be generalized to the case in which $G(z)$ is allowed to have roots in the disk, but instead there exists a root-free strip from $0$ to $z$. The main idea in this generalization is to construct a new polynomial $\phi(z)$ which embeds the disk into a strip.
Given such $\phi$, we can then approximate $G(\phi(z))$ using the  approximation via a root-free disk, since $G(\phi(z))$ is guaranteed to be root free in a disk of some radius. 
Furthermore, we can still use this approximation to estimate $G(1)$, which will encode our quantity of interest.

\begin{lemma}[Embedding a disk into a strip, Lemma 8.1 in \cite{barvinok2016approximating}]\label{lem:cir2strip}
	For $0 <\rho<1$, define 
	\begin{align}
	&\alpha =\alpha(\rho) 	= 1-e^{-\frac{1}{\rho}}, \quad \vbeta= \vbeta(\rho)=\frac{1-e^{-1-\frac{1}{\rho}}}{1-e^{-\frac{1}{\rho}}}>1,\\
	&K=K(\rho)=\left\lfloor \left(  1+\frac{1}{\rho} \right) e^{1+\frac{1}{\rho}} \right\rfloor \geq 14,\quad \sigma =\sigma(\rho) =\sum_{k=1}^K \frac{\alpha^k}{k} \text{ and }\\
	&\phi(z)=\phi_\rho(z) =\frac{1}{\sigma}\sum_{k=1}^K \frac{(\alpha z)^k}{k}.
	\label{eq:definition phi}
	\end{align}
 Then $\phi(z)$ is a polynomial of degree $K$ such that $\phi(0)=0,\phi(1)=1$, and embeds the disk of radius $\vbeta$ into the strip of width $2 \rho$, i.e., 
 \begin{align}
 	-\rho \leq \Re(\phi(z)) \leq 1+2\rho \text{ and } |\Im(\phi(z))|\leq 2\rho \text{ provided }	|z|\leq\vbeta.
 \end{align}
\end{lemma}

\begin{corollary}[Approximation using a root-free strip]\label{cor:strip} Let $G(z)$ be a polynomial of degree $n$ and suppose there exists a constant $\rho\in (0,1)$ such that $G(z)\neq 0$ for all $z$ in the strip $\cT(1,2\rho)$. Define 
$\vbeta,K,\phi(z)$ as in Lemma \ref{lem:cir2strip} where $\vbeta>1$.
 Let 
 $$F(z)\coloneqq  \LN G( \phi(z)).$$ Then $F(z)$ is analytic for $|z|\leq 1$. Moreover, consider a degree $m$ Taylor approximation of $F(z)$, 
 \begin{align}
 	P_m(z) \coloneqq  F(0) +\sum_{k=1}^m \frac{\partial^k F(z)}{\partial z^k}  |_{z=0} \frac{z^k}{k!} \label{eq:pmz}
 \end{align}
 Then, for all $|z|\leq 1$,
 \begin{align}
 	|F(z)-P_m(z)| 
 & \leq 	\frac{nK}{(m+1)\vbeta^m(\vbeta-1)}.
 \end{align}

 \end{corollary}
\begin{proof}
Recall that $\vbeta>1$.
 By Lemma \ref{lem:cir2strip} and the assumption that $G(z)\neq 0$ for all $z$ in $\cT(1,2\rho)$, we have for any $|z|\leq \vbeta$, $G(\phi(z))\neq 0$. Note that $G(\phi(z))$ is a  polynomial in $z$ of degree $nK$. Then use  Lemma \ref{lem:rootfree} w.r.t $G(\phi(z))$ and $F(z)\coloneqq \LN G(\phi(z))$ we prove the Corollary.
\end{proof}

In the above Lemmas, we have assumed $G(z)$ is a fixed polynomial, and the performance of the Taylor expansion of $F(z)$ depends on the location of roots of $G(z)$. When   $G_A(z)$ are random polynomials indexed by randomness $A$, \cite{eldar2018approximating} illustrates a way of using Jensen’s formula to  estimate the expectation of the number of roots.

For convenience of later usage, in the following we use the notation $h_A(z)$ for polynomials instead of $G_A(z)$. In later applications $h_A(z)$ will be  a rescaled version of $G_A(z)$. 
By Lemma \ref{lem:root} \cite{eldar2018approximating} 
connects the expected number of roots in a disk to the second moment of~$h_A$.

\begin{definition}[Average Sensitivity~\cite{eldar2018approximating}] \label{def:sensitivity}
Let $h_A(z)$ be a random polynomial where $A$ is sampled from some random ensembles and $h_A(0)\neq 0$. For any real number $r>0$, the stability of $h_A(z)$
at point $r$ is defined as 
\begin{align}
	k_h(r)\coloneqq  E_\theta E_A\left[ \frac{|h_A(re^{i\theta})|^2}{|h_A(0)|^2} \right]	
\end{align}
	where $E_\theta[\cdot] =\int_{\theta=0}^{2\pi}[\cdot]\frac{d \theta}{2\pi}$ is the expectation over $\theta$ from a uniform distribution over $[0,2\pi)$, and $E_A$ is the expectation over the randomness of $A$.
\end{definition}

\begin{lemma}[Proposition 8~\cite{eldar2018approximating}]\label{lem:root}
Let $h_A(z)$ be a random polynomial where $A$ is sampled from some random ensemble and $h_A(0)\neq 0$. 
Let $N_A(r)$  be the number of roots of $h_A(z)$ inside  $\cB(r)$, and $0<\nw<1/2$. Then,
	\begin{align}
		E_A \left[N_A(r-r\nw)\right] \leq \frac{1}{2\nw} \ln  k_h(r).	
	\end{align}
\end{lemma}

Lemma \ref{lem:root} bounds the expectation value of the number of roots in $\cB(r)$.  
In later sections, we will apply Lemma \ref{lem:root} to show that  our polynomial of interest has very few roots in the disk. This will allow us to find a root-free strip with high probability. 
For completeness, we provide a proof of Lemma \ref{lem:root} in Appendix \ref{appendix:barvinok}.

\section{Tensor network contraction algorithm from Barvinok's method}\label{sec:alg}

We are now ready to present our algorithm for approximate tensor network contraction. 
The algorithm is based on Barvinok's method and takes the following inputs
\begin{itemize}
    \item A tensor network $T=T(G,M)$, where $G=(V,E)$ is a graph comprising $n = |V|$ vertices and has constant degree $\kappa$.
    \item A precision parameter $\epsilon\in (0,1]$.
\end{itemize}
The goal is to approximate $\chi(T)$ with $\epsilon$-multiplicative error.
In order to achieve this, we will choose the following parameters that will enter the algorithm appropriately.
\begin{itemize}
    \item A set of non-zero complex values $\{\mu_v\}_{v\in G}$.
    We  will choose $\mu_v$ to be the mean value of the entries of the tensor $M^{[v]}$ at vertex $v \in V$.
    \item A complex value $z_{end}\neq 0$.
    \item A real value $0 < \rho <1$. This value will determine the width of the strip $\cT(1,2\rho)$ in the complex plane.
\end{itemize}

The algorithm we describe in this section is well-defined for an arbitrary tensor network. 
 In Section \ref{sec:TN} we will apply this algorithm to random 2D tensor  network whose entries have a  small positive bias and show that it succeeds with high probability.  

\subsection{The polynomial}\label{sec:poly}

To apply Barvinok's method in Corollary \ref{cor:strip}, we  map the contraction value of tensor network to a polynomial as follows. 
For each vertex $v$,
with some abuse of notations, we use $J^{[v]}$ to represent the tensor by substituting all entries in $M^{[v]}$ by $1$. 
We define
\begin{align}
    A^{[v]} \coloneqq  (\mu_v^{-1} M^{[v]}-J^{[v]})\cdot z_{end}^{-1}
\end{align}
In other word,
\begin{align}
 \mu_v^{-1}    M^{[v]} &=  J^{[v]} + z_{end} \cdot A^{[v]}\label{eq:Muv}
\end{align}
Eq.~\eqref{eq:Muv} states that we intepret the normalized version of $M^{[v]}$ as the all-one tensor $J^{[v]}$ interpolated by $A^{[v]}$. We note that we allow $z_{end}$ to be much larger than $1$.

Since the contraction value of $T(G,\{M^{[v]}\}_v)$ equals $\prod_{v} \mu_v$ times the contraction value of the normalized tensor network $T(G,\{\mu_v^{-1}M^{[v]}\}_v)$, without loss of generality, from now on we assume  that the tensor network has been normalized and 
\begin{align}
   M^{[v]} &=  J^{[v]} + z_{end} \cdot A^{[v]} \label{eq:M2A}
\end{align}

If we substitute $z_{end}$ with a variable $z$ in Eq.~\eqref{eq:M2A} for each tensor $M^{[v]}$, we will obtain a family of new tensor networks, denoted by $T_A(z)$.  
The contraction value $\chi(T_A(z))$ is a degree-$n$ polynomial in $z$. Denoting this polynomial as $g_A(z)$, we have
\begin{align}
    &g_A(z) =\chi(T_A(z)), \quad 0\leq z\leq z_{end}\\
&g_A(z_{end})=\chi(T_A(z_{end}))=\chi(T).
\end{align}

Recall that $0<\rho<1$.  Define the polynomial $\phi(z)$ as in Lemma \ref{lem:cir2strip}.
For convenience of applying Barvinok's method, we also  define $G_A(z)$ by rescaling $g(z)$,
\begin{align}
    &G_A(z)\coloneqq  g_A(z\cdot z_{end}), \quad 0\leq z\leq 1.\\
    &F_A(z)\coloneqq \LN G_A(\phi(z)).
\end{align}
 $F_A(z)$ will be analytic in the disk $\cB(1)$ if $G_A(z)$ is root-free in the strip $\cT(1,2\rho)$.

\subsection{Computing the derivatives of $g_A$}
We first note that the first few derivatives of $g_A(z)$ can be computed efficiently. 
\begin{lemma}\label{lem:dg_new}
For any integer $m$,
	 the first $m$ derivatives  $\{g_A^{(k)}(0)\}_{k=0,\ldots ,m}$ can be computed in time $O(m^2d^{\kappa m}n^{m+1})$. 
\end{lemma}
\begin{proof}
 For any subset $S\subseteq V$, we denote by $\frac{\partial T_A(z) }{\partial S}$ the tensor network which is obtained by substituting the tensor $M^{[v]}$ with $ A^{[v]}$ at every vertex $v \in S$. Using the product rule of derivatives and induction on $k$,  one can check that 
\begin{align}
	\frac{\partial^k g_A(z) }{\partial z^k} = \sum_{S \subseteq V, |S|=k} \chi\left(\frac{\partial T_A(z)}{\partial S}\right) \cdot  k! \label{eq:partial_new}
\end{align}	
In particular, note that  when $z=0$, by the definition of  $\frac{\partial T_A(0)}{\partial S}$, for any vertex $w\not \in S$, the corresponding tensor at $w$ in $\frac{\partial T_A(0)}{\partial S}$ is 
$$
M^{[v]} = J^{[v]} + 0 \cdot A^{[v]}	=J^{[v]}.
$$
As in Example \ref{ex:allone product} and Figure \ref{fig:tensor3} (f), 
we can decompose each tensor $J^{[v]} = [1,\ldots ,1]^{\otimes \kappa}$ as a product of $\kappa$ all-one vectors $[1, \ldots, 1]$. 
Then the graphical representation of  $\frac{\partial T_A(0)}{\partial S}$ consists of many disconnected sub-graphs, where each sub-graph has at most $|S|\leq m$ vertices. The contraction value $\chi(\frac{\partial T_A(0)}{\partial S})$ is the product of the contraction value of each subgraph, and can be computed in time $n\cdot O(d^{\kappa m} m)$. Here $n$ is an upper bound of the number of sub-graphs, and $O(d^{\kappa m}m)$ is the cost of directly contracting an $m$ vertices  subgraph of  a tensor network on degree-$\kappa$ graph. 
Thus by Eq.~\eqref{eq:partial_new} for any $k\leq m$, $g^{(k)}_A(0)$ can be computed in time $O(n^m d^{\kappa m}nm)$.
We conclude that the first $m$ derivatives  $\{g^{(k)}(0)\}_{k=0,\ldots ,m}$ can be computed in time $O(m^2d^{\kappa m}n^{m+1})$.
\end{proof}

In later proofs we will set $m=O\left(\ln \left(n/\epsilon\right)\right)$. When $d=poly(n)$ and  $\epsilon=O(1/poly(n))$, the cost  $O(m^2d^{4m}n^{m+1})$ of computing the $m$-th derivative is then quasi-polynomial.

Using Lemma \ref{lem:dg_new} one can efficiently compute the first few derivatives of $F_A(z)$.
\begin{lemma} \label{lem:der_F}
Assume that $\rho$ is a constant.  Then
    for any integer $m$, the first $m$ derivatives of  $\{F^{(k)}_A(0)\}_{k=1}^m$ can be computed in time
    $$O\left(m^4 d^{\kappa m} n^{m+1} + m^6 \right).$$
\end{lemma}
\begin{proof}
By Lemma \ref{lem:dg_new} and the definition of $G_A(z)$, the first $m$ derivatives $\{G_A^{(k)}(0)\}_{k=0}^m$ can be computed in time $O(m^2 d^{\kappa m} n^{m+1})$. Besides,
from Lemma \ref{lem:cir2strip} and the assumption that~$\rho$ in the definition of $\phi$ is a constant, $\phi(z)$ is a polynomial of degree $K$ where $K = K(\rho)$ is a constant, thus the first $m$ derivatives $\{\phi^{(k)}(0)\}_{k=0}^m$ can be computed in time $O(m)$. Thus by Lemma~\ref{lem:deriv_fg} in Appendix~\ref{appendix:barvinok},
 one can compute the  the  first $m$ derivatives of the composite function $G_A (\phi(z))$ at $z=0$ in time 
 $$O(m^2 d^{ \kappa m} n^{m+1}+ m^4).$$ 
 Note that $K$ is a constant and $G_A(\phi(z))$ is a polynomial of degree $nK$. 
 Then by Lemma \ref{lem:fg}, we can compute  the first $m$ derivatives $F_A^{(k)}(0)$ in time 
$$O\left(m^4 d^{\kappa m} n^{m+1} + m^6  \right).$$
\end{proof}

\subsection{The algorithm and its performance}

Our goal is to approximate $\LN G_A(1)=\LN G_A(\phi(1))$ with respect to an additive error, which will give a multiplicative approximation to $G_A(1)=\chi(T)$. The algorithm is just computing the derivatives and $P_m(z)$ in Corollary \ref{cor:strip}, that is Algorithm \ref{alg:main}.

\begin{algorithm}[H]
\caption{Barvinok($G_A,m,\rho$)}\label{alg:main}
\begin{algorithmic}[1]
\State Let $F_A(z) \coloneqq \ln G_A(\phi(z))$ with $\phi= \phi_\rho$ as defined in Eq.~(\ref{eq:definition phi}).
\State Compute the first $m$ derivatives  $\{F^{(k)}_A(0)\}_{k=0}^m$ of $F_A$ using Lemma \ref{lem:der_F}.
\State Compute $P_m(1) \coloneqq  F(0) +\sum_{k=1}^m F^{(k)}_A(0) \frac{1}{k!}$ .
\item Return $\hat{\chi}(T):=e^{P_m(1)}$.
\end{algorithmic}
\end{algorithm}

Algorithm \ref{alg:main} returns a good approximation of $\chi(T)$ with multiplicative error $\epsilon$ if $G_A(z)$ is root-free in  $\cT(1,2\rho)$ and we set $m = O(\ln(n/\epsilon))$.

\begin{theorem}\label{thm:general}
   Let $0<\rho<1$ be a constant.
   If $G_A(z)\neq 0$ for any $z$ in strip $\cT(1,2\rho)$, then for any $\epsilon$, Algorithm \ref{alg:main}   runs in time 
$$O\left(m^4 d^{\kappa m} n^{m+1} + m^6 \right).$$
Choosing  $m=O(\ln (n/\epsilon))$, Algorithm \ref{alg:main}
 outputs a value $\hat{\chi}(T)$ that approximates $\chi(T)$ with $\epsilon$-multiplicative error. That is  $$|\hat{\chi}(T)-\chi(T)|\leq \epsilon  |\chi(T)|.$$
\end{theorem}
 \begin{proof}
 As 
 in Corollary \ref{cor:strip} and Lemma \ref{lem:cir2strip}, we define  two constants $\vbeta,K$ from $\rho$.
  Set
 \begin{align}
     m\coloneqq  \frac{\ln (enK/\epsilon)- \ln (\vbeta-1)}{\ln \vbeta} = O(\ln (n/\epsilon)).
 \end{align}
  
 Define the polynomial $P_m(z)$ as  
 in Corollary~\ref{cor:strip}. Applying  Corollary~\ref{cor:strip}  to the functions $G_A(z),F_A(z)$, we have for $|z|\leq 1$,
\begin{align}
  |F_A(z)-P_m(z)|\leq \epsilon/e.  
\end{align}
One can check that $|e^x-1|\leq e |x|$ for complex $x$ where  $|x|\leq 1$. Thus for any $|z|\leq 1$,
\begin{align}
    |e^{F_A(z)}-e^{P_m(z)}| &= |e^{F_A(z)}| \cdot |1-e^{P_m(z)-F_A(z)}|\\
    &\leq |e^{F_A(z)}|\cdot e \cdot \epsilon/e.\\
    &= \epsilon\cdot |e^{F_A(z)}| \label{eq:bound_new}
\end{align}
Note that 
\begin{align}
&\chi(T)= \chi(T_A(z_{end})) = e^{F_A(1)}.
\end{align}
Since 
\begin{align}
&\hat{\chi}(T) \coloneqq  e^{P_m(1)}.
\end{align}
 Eq.~(\ref{eq:bound_new}) implies,
\begin{align}
    |\chi(T)- \hat{\chi}(T)| \leq \epsilon\cdot |\chi(T)|.
\end{align}

The runtime of the algorithm is the time for computing the polynomial $P_m(1)$, which is dominated by the time for computing $F_A^{(k)}(0)$ for $k=1,\ldots ,m$. By Lemma \ref{lem:der_F} we can compute the   the first $m$ derivatives $F_A^{(k)}(0)$ in time 
$$O\left(m^4 d^{\kappa m} n^{m+1} + m^6 \right)$$
 \end{proof}

\section{Approximating random PEPS with positive mean}\label{sec:TN}

In this section, we apply Algorithm \ref{alg:main}  to the task of approximating the contraction value of  2D tensor networks. We show that the algorithm succeeds with high probability if the tensors are drawn randomly with vanishing positive mean and intermediate
bond dimension $d\geq nc^{-1}$ for constant $c$. 
The formal statement of the result is as follows.

\begin{theorem}\label{thm:random} Suppose $d\geq nc^{-1}$ for some constant $c$.  Let $\nw$ be an arbitrary small constant satisfying $0\leq \nw \leq \min\{1/80,e^{-3c}/80\}$. Let $\epsilon\in (0,1]$ be a precision parameter. Suppose $$\mu\geq \frac{1}{d}   \frac{1}{(1-2\nw)}.$$
Then there is an algorithm $\cA$ which runs in time 
$$O\left(m^4 d^{4m} n^{m+1} + m^6 \right) \text{ for } m=O(\ln (n/\epsilon)),$$
such that with probability at least $\frac{3}{4}+\frac{1}{25}$ over the randomness of the 2D $(\mu,n,d)$-Gaussian tensor network $T$, it outputs 
 a value $\hat{\chi}(T)$ that approximates $\chi(T)$ with $\epsilon$-multiplicative error. That is  
 $$|\hat{\chi}(T)-\chi(T)|\leq \epsilon  |\chi(T)|.$$
\end{theorem}

Note that if a random variable $X\sim\cN_{\bC}(\mu,1)$ with $\mu>0$, then $\frac{1}{\mu}X \sim \cN_{\bC}(1,\frac{1}{\mu^2})$. That is 
$$\frac{1}{\mu}X = 1 + \frac{1}{\mu}Y  \text{ for  } Y\sim \cN_{\bC}(0,1).$$ 
 Thus  approximating  2D $(\mu,n,d)$-Gaussian tensor networks  with multiplicative error can be reduced to approximating  2D $(z,n,d)$-shifted-Gaussian tensor networks for $z=\frac{1}{\mu}$. 
 
 In other words, to prove Theorem \ref{thm:random} it suffices to show that one can approximate 2D $(z,n,d)$-shifted-Gaussian tensor networks $T_A(z)$ for $0\leq z \leq d\cdot (1-2\lambda)$, where 
 $$z_*=\frac{1}{\mu}\leq d\cdot (1-2\lambda)$$ gives the contraction value, that is   $\chi(T)=\mu^n \cdot \chi(T_A(z_*))$.

\subsection{Specify parameters in the algorithm}

We use the algorithm \ref{alg:main} in Section \ref{sec:alg} to prove Theorem \ref{thm:random}.  
Recall that $\lambda$ is a parameter in Theorem \ref{thm:random}.  We specify the parameters in the input of the algorithm as: 
\begin{itemize}
    \item $\kappa=4$ since the degree of a 2D lattice of periodic boundary condition is $4$.
    \item $\mu_v=1$ for all $v\in G$.
    \item $z_{end}=d\cdot (1-2\nw).$
    \item $\rho=\frac{\pi \nw^4}{4(1-2\nw)}$ for the strip $\cT(1,2\rho)$.
\end{itemize}

The key Lemma for proving Theorem \ref{thm:random}  is showing that $G_A(z)$ has no roots in $\cT(1,2\rho)$ with high probability. We use the same notations $g_A(z),G_A(z)$ as defined in Section \ref{sec:alg}.

\begin{lemma}[Root-free strip] \label{lem:rf_new}
With probability  at least $\frac{3}{4}+\frac{1}{25}$ over the randomness  of  $A$, $G_A(z)$ has no roots in $\cT(1,2\rho)$.      
\end{lemma}

Before proving Lemma~\ref{lem:rf_new} we  first prove that Theorem~\ref{thm:general} and Lemma~\ref{lem:rf_new} together imply Theorem~\ref{thm:random}.

\begin{proof}[Proof of Theorem \ref{thm:random}]
By Lemma \ref{lem:rf_new} with probability at least $\frac{3}{4}+\frac{1}{25}$ over the randomness of $A$, $G_A(z)\neq 0$ for $z\in \cT(1,2\rho)$. Then we prove  Theorem \ref{thm:random} by Theorem \ref{thm:general}.
\end{proof}

\vspace{1em}
\textbf{Rescaling polynomials.}
Recall that,
\begin{align}
    g_A(z) &= \chi(T_A(z)), \text{ for } 0\leq z\leq d\cdot (1-2\nw)\\
    G_A(z) &= \chi(T_A(z\cdot d\cdot (1-2\nw))), \text{ for } 0\leq z\leq 1.
\end{align}
For convenience, we also define
\begin{align}
    h_A(z) = \chi(T_A(zd)), \text{ for } 0\leq z\leq 1-2\nw.
\end{align}

Then Lemma \ref{lem:rf_new} is equivalent to 
\begin{lemma}[Root-free strip] \label{lem:rf}
With probability at least $\frac{3}{4}+\frac{1}{25}$ over the randomness of  $A$, $h_A(z)$ has no roots in $\cT(1-2\nw,w)$ for $w=\pi \nw^4/2$. 
\end{lemma}

The following Sections  are used to prove Lemma  \ref{lem:rf}. In   Section \ref{sec:2DIsing} we review the 2D Ising model. In Section \ref{sec:map2D}  we map the random 2D tensor network to the 2D Ising model. In Section \ref{sec:RFstrip} we analyze the partition function of the 2D Ising model and use it to find the root-free strip. Finally in Section \ref{sec:exact} we prove the exact contraction of random 2D tensor network with positive mean remains $\#\PT$-$\hard$.

\subsection{2D Ising model}\label{sec:2DIsing}
This section is a review of the 2D Ising model. Let $L_1$ and $L_2$ be two integers where
$$ n= L_1\times L_2.$$
For simplicity we assume $L_2$ is even.
Consider an $L_1\times L_2$ 2D lattice  with periodic boundary conditions, meaning that the lattice  can be embedded onto a torus. We assume the periodic boundary condition to simplify the  analysis.

Denote the lattice as $G=(V,E)$ and let  $n\coloneqq |V|$. 
At each vertex $v$, there is a spin which takes a value $s_v\in \{-1,+1\}$. 
The Hamiltonian (or energy function) of the \textit{2D Ising model} is defined to be the function $H$ mapping a spin configuration $s\coloneqq \{s_v\}_{v\in V}$ to its energy
\begin{align}
	H(s,\vJ,\vh) = -\vJ \sum_{(v,w)\in E} s_vs_w -  \vh \sum_{v\in V}s_v,\label{eq:2DIsing}
\end{align}
where $\vJ\in\bR$ is the pair-wise interaction strength and $\vh\in\bR$ quantifies the strength of an external magnetic field.
The partition function at inverse temperature $\beta$ is defined as 
\begin{align}
	\cZ(\beta,\vJ,\vh) &=\sum_{s\in \{\pm 1\}^n}	\exp(-\beta H(s,\vJ,\vh)) \\
		&=\sum_{s \in \{\pm 1\}^n} \prod_{(v,w)\in E} \exp(\beta \vJ s_vs_w)\cdot \prod_{v\in V}\exp(\beta \vh s_v)
\end{align}
 It is well known that when there is no external magnetic field, that is $\vh=0$, the partition function
 of the 2D Ising model with periodic boundary has a closed form.  
 In the thermodynamic limit, this formula is
 known as Onsager's solution~\cite{onsager1944crystal}. For a finite-size lattice, a refined formula has been given by Kaufman~\cite{kaufman1949crystal}, which is summarized in Lemma \ref{lem:2DIsing}.
 There is no closed form formula for the partition function $\cZ(\beta,\vJ,\vh)$ when $\vh\neq 0$.

\begin{lemma}[\cite{kaufman1949crystal}]\label{lem:2DIsing} % eqs(39)(51)(68)(7) in \cite{kaufman1949crystal},
The partition function of the 2D Ising model on an $L_1\times L_2$ lattice with periodic boundary conditions and zero magnetic fields is given by 
\begin{align}
    \cZ(\beta,\vJ,0) = &\frac{1}{2} \left( 2 \sinh 2\beta \vJ
  \right)^{L_1L_2/2} \times \left\{ \prod_{r=1}^{L_2} \left( 2 \cosh \frac{L_1}{2}\gamma_{2r}  \right)  + 
  \prod_{r=1}^{L_2} \left( 2 \sinh \frac{L_1}{2}\gamma_{2r}  \right) \right. \\
  &+\left.
  \prod_{r=1}^{L_2} \left( 2 \cosh \frac{L_1}{2}\gamma_{2r-1}  \right) +
  \prod_{r=1}^{L_2} \left( 2 \sinh \frac{L_1}{2}\gamma_{2r-1}  \right)
  \right\}
\end{align}
where for $j=1,\ldots ,2L_2$, we define
\begin{align}
    &\cosh \gamma_{j} \coloneqq  \cosh 2H^* \cdot\cosh 2\beta \vJ - \sinh 2H^* \cdot\sinh 2\beta \vJ\cdot \cos(j\pi/L_2),\label{eq:gamma_j}\\
    &\tanh H^*\coloneqq \exp(-2\beta \vJ).
\end{align}
\end{lemma}
Notice that Eq.~\eqref{eq:gamma_j} does not specify the sign of $\gamma_j$. 
Since here we are only interested in an upper bound of $ \cZ(\beta,\vJ,0)$, we can just assume that $\gamma_j$ has a positive sign.\footnote{ For readers who are interested in numerically verifying  Lemma~\ref{lem:2DIsing}, the sign of $\gamma_j$  influences the value of $\sinh \frac{L_1}{2}\gamma_{j}$ and thus the value of $\cZ(\beta,\vJ,0)$.
The sign of $\gamma_j$ is explained in Remark 15 and Figure 3 of \cite{kaufman1949crystal}: $\gamma_j\geq 0$ for all $j\neq 2n$;  but $\gamma_{2n}$ is negative if $\beta \vJ< H_c$, and is non-negative if $\beta \vJ\geq H_c$, where $H_c$ is the critical point which is approximately 0.4407.  A remark is \cite{kaufman1949crystal} denotes our $\gamma_{2n}$ as $\gamma_0$.}

\subsection{Mapping random tensor networks to the Ising model}\label{sec:map2D}
 
 In this section, we estimate $E_A |h_A(z)|^2$ by mapping it to the partition function of a 2D Ising model. 
To this end,  observe that choosing $\beta,\vJ$ such that  $\exp(\beta \vJ) = d^2$ and $\exp(-\beta \vJ) = d \sqrt d$, we can write 
\begin{align}
 	\cZ(\beta, \vJ,0) = \sum_{s \in \{\pm 1 \}^n} R(s)
 \end{align} 
in terms of a function
\begin{align}
	 R(s) &= \prod_{(v,w)\in E} r_{vw}(s),
\end{align}
with weights  
\begin{align}
	&r_{vw}(s) = \left\{
	\begin{aligned}
		&d\sqrt{d}, \text{ if } s_u\neq s_v,\\
	& d^2, 	\text{ if } s_u= s_v.
	\end{aligned}
	\right.
\end{align}

We then show the following lemma.
  
\begin{lemma}\label{lem:EZ} For $z\in \bC$, $|z|>0$, set  $\beta,\vJ,\vh$ in the 2D Ising model to satisfy
   	 \begin{align}
 	&\beta \vJ = \frac{\ln d}{4},\quad \beta \vh = \ln |z|.
 	 \end{align}
 Then we have that over the randomness of $A$, we have that 
 \begin{align}
 E_A |h_A(z)|^2 &= d^{7n/2}|z|^n \cZ(\beta,\vJ,\vh)\\
 				&= \sum_{ s\in\{\pm 1\}^n} R(s)|z|^{2|s|}.
 \end{align}
 \end{lemma}
Note that in Lemma \ref{lem:EZ}, the 2D Ising model has a non-zero magnetic field $\vh\neq 0$, thus the closed form formula for 2D Ising model without magnetic fields (Lemma \ref{lem:2DIsing}) does not directly apply.  

In the remainder of this section, we prove Lemma \ref{lem:EZ}. 
Here we use the techniques of mapping random instances to  classical statistical mechanical models, which are widely used in the physics literature for 
studying phase transitions~\cite{bao2020theory,skinner2019measurement,yang2022entanglement,levy2021entanglement}. 
This section will heavily use the graphical representations of tensor networks, which was explained in Section \ref{sec:Notation}.

Recall that in the 2D  tensor network $T_A(zd)$, for each vertex $v$, the tensor $M^{[v]}$ can be written as 
\begin{align}
M^{[v]} = J^{[v]} + zd \cdot A^{[v]}.	
\end{align}
In the following Lemma \ref{lem:prod},
we first compute the expectation of the product of tensor $M^{[v]}$ and its conjugate, that is $M^{[v]}\otimes \overline{M^{[v]}}$.
Evaluating this average will allow us to compute $E_A|h_A(z)|^2 $, since $h_A(z)\cdot\overline{h_A(z)}$ is the product of two 2D tensor networks where, for each vertex we can pair the tensors as $M^{[v]} \otimes \overline{M^{[v]}}$, as we will explain in detail in the proof of Lemma~\ref{lem:EZ}.

\begin{lemma}\label{lem:prod}
Define the delta tensor $\dt$ to be a tensor of rank $2$ with free edges $i,i'$ and bond dimension $d$, where $(\dt)_{ii'}=\delta_{ii'}$. 	As in Figure \ref{fig:map2D2} , we have  
	\begin{align}
		& E_A \left[ A^{[v]}\otimes \overline{A^{[v]}} \right] =  \dt	\otimes \dt	\otimes \dt	\otimes \dt \coloneqq    \left(\dt\right)^{\otimes 4}	,\label{eq:a} \\	
		& E_A \left[ M^{[v]}\otimes \overline{M^{[v]}} \right] = J^{[v]}\otimes J^{[v]} + |z|^2 (d^{1/2} \cdot \dt)^{\otimes 4}.\label{eq:b}
	\end{align}
  where in Figure \ref{fig:map2D2} (a) we use $\Box$ to represent the vertex for a delta tensor $\dt$, and in Figure \ref{fig:map2D2} (b) we use $\bullet$ to represent the vertex for the  tensor $[1,1,1\ldots ,1]$.
   \begin{figure}[H]
      \centering      \includegraphics[width=\textwidth]{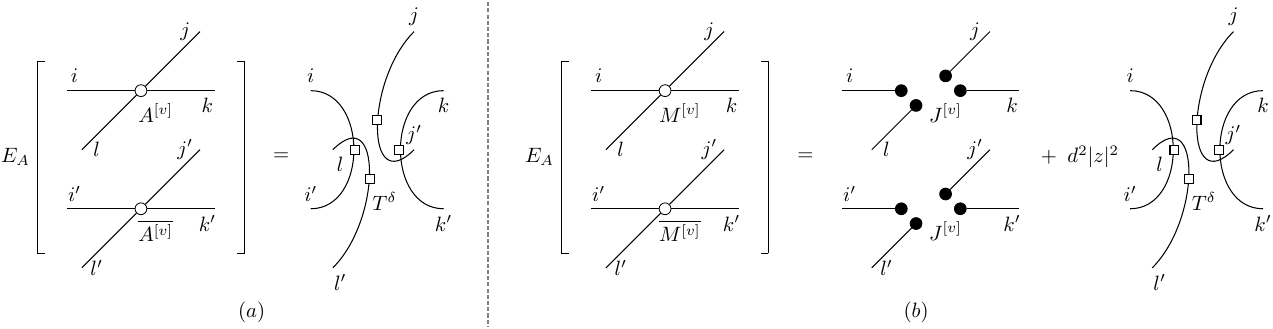}
      \caption{ (a) The expectation of the product of $A^{[v]}$ and its conjugate is a product of delta tensors. (b) The expectation of the product  $M^{[v]} \otimes \overline{ M^{[v]}}$ decomposes into a linear combination of delta tensors and a product of rank-$1$ tensors.}
      \label{fig:map2D2}
  \end{figure}

\end{lemma}
\begin{proof}
	As in	Figure \ref{fig:map2D2} (a), we label the free edges of the first copy of $A^{[v]}$ by $i,j,k,l$, and the free edges of the second copy by $i',j',k',l'$. By definition, all entries of $A^{[v]}$ are sampled independently from $\cN_{\bC}(0,1)$. Thus, we have that
\begin{align}
 	\left(E_A \left[ A^{[v]}\otimes \overline{A^{[v]}} \right]\right)_{ijkl,i'j'k'l'}=\delta_{ii'}\cdot \delta_{jj'}\cdot \delta_{kk'}\cdot \delta_{ll'},
\end{align}
which proves Eq.~\eqref{eq:a}.
To prove Eq.~\eqref{eq:b}, it suffices to notice that 
\begin{align}
	& 	M^{[v]} = J^{[v]} + zd \cdot A^{[v]}.\\
	& E_A\left[ J^{[v]}\otimes \overline{A^{[v]}}\right] =  E_A\left[ A^{[v]}\otimes \overline{J^{[v]}}\right] =0.
\end{align}
\end{proof}

\begin{proof}[Proof of Lemma \ref{lem:EZ}]
\begin{figure}[H]
      \centering
      \includegraphics[width=0.7\textwidth]{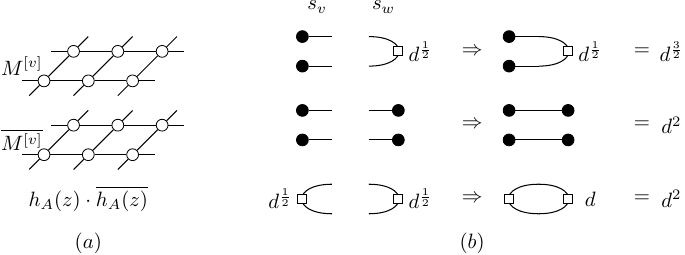}
      \caption{(a) Pair the tensors $M^{[v]}$ and its conjugate in $E_A |g_A(z)|^2.$  (b) The tensors with respect to edge $(v,w)$ in $T(s)$. From top to button, the value of $(s_v,s_w)$ are $(1,-1),(1,1),(-1,-1)$ respectively. }  
      \label{fig:mapparfun}
  \end{figure}
 As in Figure \ref{fig:mapparfun} (a), 
 $h_A(z)\cdot \overline{h_A(z)}$ is the product of two 2D tensor network, where for each vertex $v$, we can pair the tensors as $M^{[v]}\otimes \overline{M^{[v]}}$.	
	 By Lemma \ref{lem:prod} we know that 
  \begin{align}
      N^{[v]}\coloneqq E_A \left[ M^{[v]}\otimes \overline{M^{[v]}} \right] = J^{[v]}\otimes J^{[v]} + |z|^2 (d^{1/2} \cdot \dt)^{\otimes 4}.\label{eq:T(N)}
  \end{align}
 Define a new 2D tensor network $T(N)$ where at each vertex $v$ the tensor is $N^{[v]}$.   
 Notice that $E_A |h_A(z)|^2=\chi(T(N))$ since
	 \begin{align}
	 E_A |h_A(z)|^2 &= E_A\left(\sum_{c\in[d]^{|E|}} \prod_{v\in V} M_c^{[v]}\right)  \left(\sum_{c\in[d]^{|E|}} \prod_{v\in V} \overline{M_c^{[v]}}\right)\\
  &= \sum_{c,c'\in [d]^{|E|}} \prod_{v\in V} E_A \left(M^{[v]} \otimes \overline{M^{[v]}}\right)_{c,c'}\\
    &= \chi(T(N)).
	 \end{align}

  We  map  $\chi(T(N))$ to the partition function of the 2D Ising model as follows.   For any configuration $s\coloneqq \{s_v\}_{v\in V}$, $s_v\in \{\pm 1\}$, construct a new 2D tensor network $T(s)$ as follows:  
  \begin{itemize}
      \item  If $s_v=1$
	 we set the tensor on $v$ to be $ J^{[v]}\otimes J^{[v]}$; 
    \item  If $s_v=-1$ we set the tensor on $v$ to be $(d^{1/2} \cdot \dt)^{\otimes 4}$. 
  \end{itemize}
	 As in the top figure in Figure \ref{fig:mapparfun} (b), for an edge $(v,w)$, if $(s_v,s_w)=(1,-1)$, then the edge $(v,w)$ contributes a scalar factor as $d^{\frac{3}{2}}$ to $T(s)$, which is the contraction value of the tensor $J^{[v]}\otimes J^{[v]}$ and the tensor $(d^{1/2} \cdot \dt)^{\otimes 4}$. Similarly, as in  Figure \ref{fig:mapparfun} (b), if $(s_v,s_w)=(1,1)$ or $(-1,-1)$, the edge $(v,w)$ contributes a scalar factor as $d^2$. Thus 
 
  one can check that the contraction value of $T(s)$ is given by  $R(s)$.   
  
  With an arbitrary ordering of the $n$ vertices, we write the  configuration $s=\{s_v\}_{v\in V}$ as a vector $s\in\{\pm 1\}^n$. 
	 Based on Eq.~\eqref{eq:T(N)},
  one can compute $\chi(T(N))$ by expanding $N^{[v]}$, that is for any configuration $s$, 
  \begin{itemize}
      \item  We use $s_v=1$ to represent choosing $J^{[v]}\otimes J^{[v]}$,
      \item We use $s_v=-1$  for $|z|^2\cdot (d^{1/2} \cdot \dt)^{\otimes 4}$.
  \end{itemize}
  Then define $|s|$ to be the number of $-1$ in $s$, we have
	 \begin{align}
	 		\chi(T(N))
    =  \sum_{s\in \{\pm 1\}^n} \chi(T(s)) |z|^{2|s|}= 
    \sum_{s\in \{\pm 1\}^n}  R(s) |z|^{2|s|}
	 \end{align}
One can check that setting $\beta \vJ = \frac{\ln d}{4}$ and
 	$\beta \vh = \ln |z|$,
for any $s\in\{\pm 1\}^n$, we have
\begin{align}
	d^{7n/2}	 |z|^n \exp(-\beta H(s,\vJ,\vh)) =  R(s) |z|^{2|s|}.
\end{align}
\end{proof}

 \subsection{Finding a root-free strip } \label{sec:RFstrip}
  In this section, we show that one can efficiently find a root-free strip with high probability.  
  In particular, we will  bound  $E_A[|h_A(z)|^2]$ and use Lemma \ref{lem:root}. 

 Recall that we consider a 2D  lattice with periodic boundary conditions, where the 2D lattice has size $n=L_1\times L_2$ and $L_2$ is even.  The 2D Ising model and $R(s)$ are defined in Section \ref{sec:2DIsing}.

The exact formula for the partition function $\cZ(\beta,\vJ,0)$ in Lemma \ref{lem:2DIsing} is intimidating.
We upper bound $\cZ(\beta,\vJ,0)$ by a  simpler formula. Then we will use this formula to bound $E_A[|h_A(z)|^2]$.

\begin{lemma}[Bound on the partition function with no magnetic field]\label{lem:leq}
If $\beta \vJ=1/4 \cdot \ln d$ and $d\geq 3$, we have that 
\begin{align}
    2d^{\frac{n}{2}}\leq \cZ(\beta,\vJ,0)\leq 2d^{\frac{n}{2}}\left(1+\frac{3}{d}\right)^n
\end{align}
\begin{proof}
   Here we use the same notation as in Lemma \ref{lem:2DIsing}. By the definition of the partition function, and the fact that there are $2n$ edges in the 2D square lattice with periodic boundary conditions, we have
\begin{align}
\label{eq:lower bound}
    \cZ(\beta,\vJ,0)\geq \sum_{s=00..0\text{ or }11\ldots 1} \prod_{(v,w)\in E} \exp(\beta \vJ s_vs_w) = 2d^{n/2}.
\end{align}
Now, note that we can rewrite the definitions in Lemma \ref{lem:2DIsing}, for any $j$ as  
   \begin{align}
       &2\cosh \gamma_j = \cosh 2H^*\cdot 2\cosh 2\beta \vJ \cdot \left( 1- \alpha  \cos (j\pi/L_2)\right), \label{eq:2gammaj},
   \end{align}
   where $\alpha\coloneqq \tanh 2H^*\cdot \tanh 2\beta \vJ \in (0,1)$.
  Since $L_2$ is even, we have
   \begin{align}
       \prod_{r=1}^{L_2} (1-\alpha \cos (2r\pi/L_2)) &=\prod_{r=1}^{L_2/2} \left(1-\alpha \cos (2r\pi/L_2)\right) \left(1-\alpha \cos (2(r+\frac{L_2}{2})\pi/L_2) \right)\\
       &= \prod_{r=1}^{L_2/2} \left(1-\alpha \cos (2r\pi/L_2)\right) \left(1 + \alpha \cos (2r\pi/L_2) \right)\\
       &= \prod_{r=1}^{L_2/2} \left(1-\alpha^2 \cos^2 (2r\pi/L_2)\right) \\
       &\leq 1. \label{eq:leq1}
   \end{align}
By the definition of $\cosh$,  $2\cosh kx\leq (2\cosh x)^k$ for any  $x$ and any integer $k\geq 0$,  we have that
   \begin{align}
       \prod_{r=1}^{L_2} \left(2\cosh \frac{L_1}{2} \gamma_{2r}\right) &\leq \left(  \prod_{r=1}^{L_2} 2\cosh \gamma_{2r}\right)^{L_1/2}\\
       &= \left(\cosh 2H^*\cdot 2\cosh 2\beta \vJ\right)^{L_1L_2/2} \left(\prod_{r=1}^{L_2} (1-\alpha \cos (2r\pi/L_2))\right)^{L_1/2}\\
       &\leq \left(\cosh 2H^*\cdot 2\cosh 2\beta \vJ\right)^{L_1L_2/2}.
   \end{align}
   where the second equality comes from Eq.~(\ref{eq:2gammaj}) and the last inequality comes from Eq.~(\ref{eq:leq1}).
Similarly we can get the same upper bound for 
$$
\prod_{r=1}^{L_2} \left(2\cosh \frac{L_1}{2} \gamma_{2r-1}\right), \text{ and } \prod_{r=1}^{L_2} \left(2\sinh \frac{L_1}{2} \gamma_{2r}\right), \prod_{r=1}^{L_2} \left(2\sinh \frac{L_1}{2} \gamma_{2r-1}\right). 
$$
where we get the bound for the last two terms by 
$|\sinh x|\leq \cosh x$.
 Besides, from $\tanh H^* = \exp(-2\beta \vJ)$ we have
   \begin{align}
       &\exp(2H^*) = \frac{\exp(2\beta \vJ)+1}{\exp(2\beta \vJ) -1 } = \frac{\sqrt{d}+1}{\sqrt{d}-1}\\
         &\cosh 2H^* =\frac{1}{2}\left( \frac{\sqrt{d}+1}{\sqrt{d}-1} + \frac{\sqrt{d}-1}{\sqrt{d}+1} \right) \leq 1 +\frac{3}{d}, \text{ for } d\geq 3.
   \end{align}
We estimate
\begin{align}
      &2 \sinh 2\beta \vJ =  \sqrt{d}-\frac{1}{\sqrt{d}} \leq \sqrt{d}.\\ &2 \cosh 2\beta \vJ =  \sqrt{d}+\frac{1}{\sqrt{d}} 
\end{align}
Thus by Lemma \ref{lem:2DIsing}, we finally conclude that
\begin{align}
    \cZ(\beta,\vJ,0) &\leq \frac{1}{2}\cdot d^{\frac{n}{4}} \cdot 4 \cdot \left(1+\frac{3}{d}\right)^{\frac{n}{2}} d^{\frac{n}{4}} \left(1+\frac 1 d \right)^{\frac{n}{2}}\\
    &\leq 2 d^{\frac{n}{2}} \left(1+\frac{3}{d}\right)^n .
\end{align}
\end{proof}
\end{lemma}
Using Lemma \ref{lem:leq} we can now estimate  $\sum_s R(s)$.
\begin{lemma}\label{lem:total}
	\begin{align}
		\sum_{s \in \{\pm 1\}^n} R(s) \leq 2 d^{4n} (1+\frac{3}{d})^n	
	\end{align}
\end{lemma}
\begin{proof}
Since the 2D lattice we consider has $n$ vertices and satisfies the periodic boundary condition, there are in total $2n$ edges.
	Set $z=1, \beta \vJ = 1/4\cdot \ln d$, 
 	$\beta \vh = \ln |z|=0$, (that is $\vh=0$).  By Lemma \ref{lem:EZ}  and Lemma \ref{lem:leq} we have that 
	\begin{align}
		\sum_{s\in\{\pm1\}} R(s) &= d^{7n/2} \cZ(\beta,\vJ,0)
		\leq 2d^{4n}\left(1+\frac{3}{d}\right)^n.
	\end{align}
\end{proof}

Then we  estimate  $E_A |h_A(z)|^2$ for small $z$ and for $z\leq 1$. For completeness we also give a lower bound on $E_A |h_A(z)|^2$ in Lemma \ref{lem:variance} (c). (c) will not be used in other proofs.

\begin{lemma}\label{lem:variance} Let $c$ and $\rho$ be two constants where $0<\rho<1$.  Assume $d\geq  n c^{-1}$.We have 
  	\begin{itemize} 
	\item[(a)] For $|z|\leq \rho$,  $E_A |h_A(z)|^2\leq d^{4n} (1+2\rho^2 e^{3c}).$
	\item[(b)] For $|z|\leq 1,$ $E_A |h_A(z)|^2\leq d^{4n} \cdot 2 e^{3c}.$  
 \item[(c)] For any $z$, $E_A |h_A(z)|^2\geq d^{4n} (1+ \frac{|z|^2}{d^4} )^n$. 
\end{itemize}	
\end{lemma}
 \begin{proof} Note that when $s=0\ldots 0$, $R(s)=d^{4n}$. 
 For (a), since $|z|\leq \rho$, by Lemma \ref{lem:EZ} we have
 	\begin{align}
 		E_A |h_A(z)|^2 &= d^{4n}\cdot |z|^{0}
   + \sum_{s\in\{\pm1\}^n:|s|\geq 1} R(s) |z|^{2s}\\
   &\leq d^{4n}
   + \rho^2 \sum_{s\in\{\pm1\}^n} R(s) \\
   &\leq d^{4n} + \rho^2 \cdot 2 d^{4n} \left(1+\frac{3}{d}\right)^n \\
   &\leq d^{4n} (1+2\rho^2 e^{3c}).
 	\end{align}
where the second inequality comes from $|z|\leq\rho<1$ and $R(s)\geq 0, \forall s$;  the third inequality comes from Lemma \ref{lem:total}; and the last inequality comes from $d\geq  n\cdot c^{-1}$.

For (b),  by Lemma \ref{lem:EZ} and Lemma \ref{lem:total}  we have for $|z|\leq 1$,
\begin{align}
    E_A|h_A(z)|^2 &\leq \sum_{s\in\{\pm1\}^n} R(s)\cdot 1 \\
    &\leq d^{4n}2\left(1+\frac{3}{d}\right)^n\\
    &\leq  d^{4n} \cdot 2 e^{3c}.
\end{align}

For (c), note that  for $|s|=k$, since there are at most $4k$ edges $(v,w)$ in $R(s)$ which take values $r_{vw}(s)=d\sqrt{d}$, we must have 
$R(s)\geq d^{4n}/\sqrt{d}^{4k}$. Thus
    \begin{align}
        E_A|h_A(z)|^2 &= \sum_{k=0}^n \sum_{s\in\{\pm1\}^n:|s|=k} R(s) |z|^{2k}.\\
        & \geq \sum_{k=0}^n \binom{n}{k} d^{4n} \left(\frac{|z|^2}{d^2}\right)^k\\
        & = d^{4n} \left(1+ \frac{|z|^2}{d^2} \right)^n.
    \end{align}
 \end{proof}

Recall that $N_A(r)$ is the number of roots of $h_A(z)$ inside the disk  $\cB(r)$. With Lemma~\ref{lem:variance}, 
 we can estimate $N_A(r)$  by using Lemma~\ref{lem:root}.

\begin{corollary}\label{cor:14}
Suppose $d\geq nc^{-1}$ for some constant $c$.  Let $\nw$ be an arbitrary small constant satisfying $0\leq \nw \leq \min\{1/80,e^{-3c}/80\}$.
 We have 
	\begin{align}
		Pr_A\left[ N_A(\nw)	= 0 \quad \& \quad N_A(1-\nw) \leq \frac{1}{\nw^2}\right] \geq 4/5.
	\end{align}
\end{corollary}
\begin{proof}
Note that 
$$h_A(0)=d^{2n}.$$
We have
	\begin{align}
			E_A [N_A(\nw)] & \leq E_A \left[
   N_A\left(2\nw -2\nw\cdot \nw/2\right) \right]\\
    &\leq \frac{1}{\nw} \ln E_\theta E_A \frac{|h_A(2\nw\cdot e^{i\theta})|^2}{|h_A(0)|^2}\\
   &\leq  \frac{1}{\nw} \cdot  \ln (1+8\nw^2 e^{3c})\\
   &\leq 8\nw\cdot e^{3c};
   \end{align}
   where the first inequality comes from 
  the fact that the disk of radius $2\nw-2\nw\cdot \nw/2$ contains the disk of radius $ \nw$; the second inequality comes from 
  Lemma \ref{lem:root}; and the third inequality comes from Lemma \ref{lem:variance} (a) by setting $\rho=2\lambda$.
 
 Similarly from Lemma \ref{lem:root} and Lemma \ref{lem:variance} (b)
 we get
   \begin{align}
			 E_A\left[N_A\left(1-\nw\right)\right] &\leq \frac{1}{2\nw} \ln E_\theta E_A \frac{|h_A(e^{i\theta})|^2}{|h_A(0)|^2}\\
    &\leq \frac{1}{2\nw} \cdot
   \ln (2e^{3c}).
	\end{align}
Then by Markov inequality, we have
	\begin{align}
		 Pr_A\left[ N_A(\nw)\geq 1 \right] &\leq E_A \left[N_A(\nw)\right] \\
		&\leq 8\nw\cdot e^{3c}\\
		&\leq \frac{1}{10}.
		\end{align}
	\begin{align}
		Pr_A\left[ N_A\left(1-\nw\right)\geq \frac{1}{\nw^2}\right] &\leq \nw^2\cdot E_A\left[N_A(1-\nw)\right]\\
		&\leq \nw/2 \cdot \ln (2e^{3c}) \\
		&\leq \frac{1}{10}.
	\end{align}
Then we get the desired result by a union bound.
\end{proof}

Finally we prove Lemma \ref{lem:rf}.

\begin{proof}[Proof of Lemma \ref{lem:rf}]
	Let 
	\begin{align}
		M\coloneqq 1/\nw^3,\quad	\theta \coloneqq  2\pi/M, \quad w\coloneqq \pi\nw^4/2.
	\end{align}
 For simplicity we assume that $1/\nw$ is an integer. Consider the disk $\cB(1-\nw)$, that is the disk centered at $0$ and of radius $1-\nw$. As in Figure \ref{fig:disjoint_strip}, we divide $\cB(1-\nw)$ into $M$ disjoint circular sectors, where for each sector the central angle is $\theta$. Insider each sector which is indexed by $k\in [M]$, we consider a strip of width $w$, that is
 $$\cT_k\coloneqq  \cT((1-2\nw)e^{ik\theta},w).
 $$
 Note that since $\sin x \geq x/2$ for $0\leq x\leq \pi/3$, we have
 \begin{align}
     \nw \sin\left(\theta/2\right) \geq \nw  \theta/4 =\pi \nw^4/2=w.
 \end{align}
 Thus all the $M$ strips $\{\cT_k\}_{k=1}^M$ are disjoint outside $\cB(\nw)$. Besides, one can check that the end part of the strip $\cT_k$ is inside the $k$-th sector by noticing 
   $$0\leq \nw\leq 1/80 \Rightarrow   
 (1-2\nw)\tan \frac{\theta}{2}\geq w.$$

  Denote $S_{good}$ as 
	  the set of tensors $A$ which have no roots in the disk of radius $\lambda$ and few roots in the disk of radius $1- \lambda$: 
	  $$S_{good}\coloneqq  \{A: \quad N_A(\nw)	= 0\,\,\&\, \, N_A(1-\nw) \leq \frac{1}{\nw^2}\}.$$ 
	By Corollary \ref{cor:14}, we know that
	\begin{align}
		Pr_A\left[ A\in S_{good}\right] \geq 4/5.
	\end{align}
	In the following we argue that $S_{good}$ can be further partitioned into disjoint subsets, where in each subset, with probability at least $(1-\nw)$ over the randomness of $A$ there are no roots in $\cT(1-2\nw,w)$. 
   
 To this end, first we observe that for any
 $k\in [M]$, by the definition of $h_A(\cdot)$, we have
 \begin{align}
    h_{e^{ik\theta}A  }(z) =  h_A(e^{ik\theta}z).\label{eq:equi}
 \end{align}
Note that $e^{ik\theta}z$ is just a rotation of $z$ in the complex plane. By the rotational symmetry of  disks, the roots of $ h_{e^{ik\theta}A  }(z) =  h_A(e^{ik\theta}z)$  are simply rotated compared to the roots of $h_{A}(z)$.
 Thus if $A\in S_{good}$, then so is $e^{ik\theta}A$. 
  
  Next, we partition $S_{good}$ into disjoint subsets in the way that  $A,A'$ are in the same subset iff there exists $k\in [M]$ such that $A'= e^{ik\theta} A$. 
  For convenience, for each subset we fix an arbitrary $A$ as the representative and write the subset  as 
  $$S_{good}(A)=\{e^{ik\theta} A | \, k\in[M]\}.$$

 By the definition of $S_{good}$, for any $A\in S_{good}$, there are no roots in $\cB(\nw)$ and there  are at most $1/\nw^2$ roots in $\cB(1-\nw)$. Since the $M$ tubes $\{\cT_k\}_{k=1}^M$ are disjoint outside $\cB(\nw)$,  there is at most a $\frac{1/\nw^2}{M}=\nw$ fraction of the $M$ tubes which contains roots of $h_A(z)$. 
 Further, recall that 
 $$h_{e^{ik\theta}A  }(z) =  h_A(e^{ik\theta}z),$$ 
 and thus the tube $\cT_0$ with respect to to $e^{ik\theta}A$ corresponds to  the tube $\cT_k$ with respect to $A$. 
 Hence, there is at most a $\nw$ fraction of $A'\in S_{good}(A)$ such that the corresponding strip $\cT_0$ contains roots of $h_A(z)$.

  \begin{figure}[H]
      \centering
\includegraphics[width=0.5\textwidth]{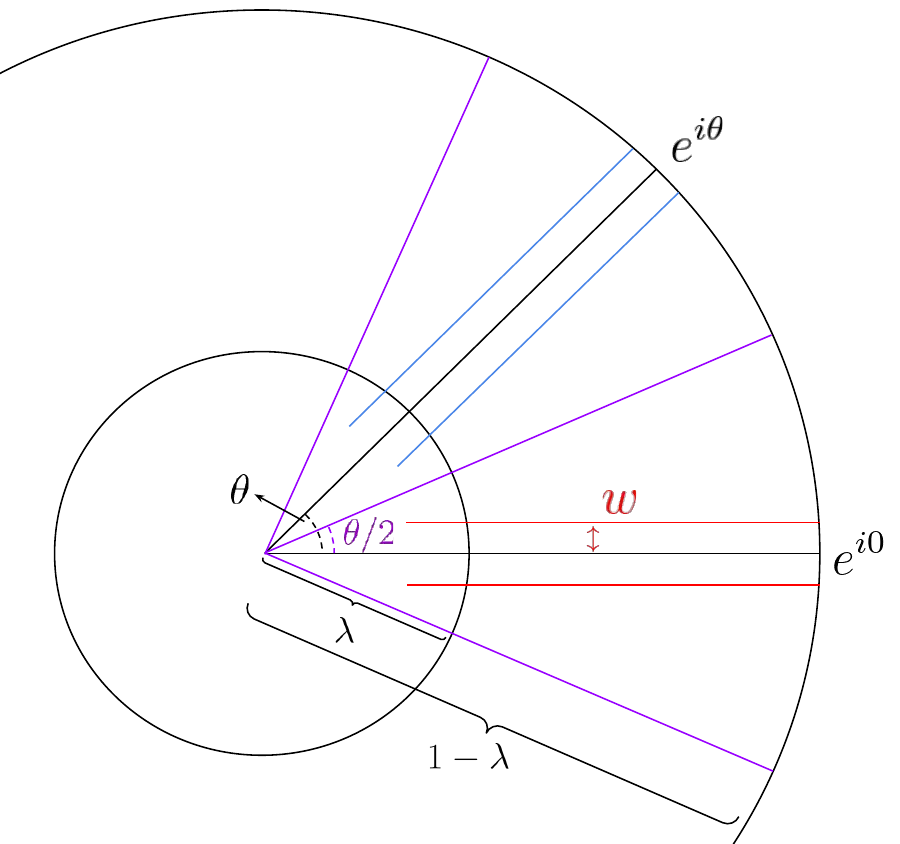}
      \caption{ The radius of the small disk and the big disk is $\lambda$ and $1-\lambda$ respectively. We divide
     the big disk $\cB(1-\lambda)$ into $M$ circular sectors. In  each sector we choose a strip of width $w$. The first strip $\cT_0$ starts from $-w$ and ends at $1-2\lambda+w$.  Other strips  are rotations of $\cT_0$.
     All the strips are disjoint outside the small disk $\cB(\lambda)$. 
     }
\label{fig:disjoint_strip}
  \end{figure}

In summary, we conclude that the fraction of $A$ such that there are no roots in $\cT(1-2\nw,w)$ is greater than
$$
\frac{4}{5}\cdot (1-\nw) \geq \frac{3}{4} +\frac{1}{25}.$$
 	\end{proof}

\subsection{ $\# \PT$-hardness of exact contraction}\label{sec:exact}

Finally we prove that the exact contraction of  the random 2D tensor network with a positive mean remains $\#\PT$-$\hard$. The proof is a simple adaption of Theorem 1 and Theorem 3 in \cite{haferkamp2020contracting}. For completeness, we put a proof in Appendix  \ref{appendix:exact-case}.

To make the statement rigorous, here we consider the finite precision approximation of the Gaussian distribution, denoted as $\overline{\cN}_{\bC}(\mu,\sigma^2)$, where each sample can be represented by finite bits instead of being an arbitrary real or complex number. For example here we set the   $\overline{\cN}_{\bC}(\mu,\sigma^2)$ to be the distribution where each sample $z\sim \overline{\cN}_{\bC}(\mu,\sigma^2)$ is obtained by:  firstly sample $y$ according to Gaussian distribution $ \cN_{\bC}(\mu,\sigma^2)$, then set $z$ to be the value by rounding $y$ to $n^2$ bits. $\overline{\cN}_{\bC}(\mu,\sigma^2)$ behaves similarly as $\cN_{\bC}(\mu,\sigma^2)$ but makes the statements of exact contraction and proofs more rigorous. 

Accordingly,  we consider finite precision 2D $(\mu,n,d)$-Gaussian tensor network instead of 2D $(\mu,n,d)$-Gaussian tensor network, where we substitute $\cN_{\bC}(\mu,\sigma^2)$ by $\overline{\cN}_{\bC}(\mu,\sigma^2)$.

\begin{theorem}[$\#\PT$-$\hard$]\label{thm:exactH}
For any $\mu\in [0,poly(n)]$,  $n\geq 25$ and $d=O(poly(n))$, if there exists an algorithm $\cA$ which runs in $poly(n)$ time and with probability at least $\frac{3}{4}+\frac{1}{n}$ over the randomness of the finite precision 2D $(\mu,n,d)$-Gaussian tensor network
    $T$, it outputs the exact value of $\chi(T)$, then there exists an algorithm which runs in randomized $poly(n)$ time  and solves $\#\PT$-\Complete\ problems.
\end{theorem}

\section{Approximating arbitrary positive tensor networks}\label{sec:positive}

In previous sections we have considered  approximating random tensor networks. In this section we move to the task of contracting a fixed  tensor network.

For a general tensor network $T=T(G,M)$, computing the contraction value $\chi(T)$ \textit{exactly} is known to be $\#\PT$-$\hard$ \cite{schuch2007computational}. 
On the other hand, Arad and Landau~\cite{arad2010quantum} proved that \textit{approximating} $\chi(T)$ up to an inverse polynomial additive error in the matrix \textit{$2$-norm} is $\BQP$-$\Complete$.

 In this section, we focus on \textit{positive tensor network} $T=T(G,M)$. 
 These are defined by tensors $\{M^{[v]}\}_v$ the entries of which are all non-negative. The main part of this section will establish  that when $T$ is a positive tensor network,  approximating 
$\chi(T)$ up to an inverse polynomial additive error in the matrix \textit{$1$-norm} is $\BPP$-$\Complete$.
Then, in Section \ref{sec:StoqMA} we give a short proof showing that approximating positive tensor network with inverse-poly multiplicative error is at least $\StoqMA$-$\hard$. Section \ref{sec:StoqMA} is self-contained and can be read independently.  
We first review the swallowing algorithm for tensor network contraction. Then we  explain 
Arad and Landau's $\BQP$-$\Completeness$ result and  our $\BPP$-$\Completeness$ result, which are both based on the swallowing algorithm.

\subsection{A swallowing algorithm and notations}
\label{sec:SWA}

Recall that in Section \ref{sec:Notation} we have introduced two operations  on tensor networks, taking their product and contraction. 
Given a tensor network $T=T(G,M)$, 
the swallowing algorithm (Algorithm \ref{algo:swa}) is a standard method to exactly compute the contraction value $\chi(T)$, by contracting edges of tensors $\{M^{[v]}\}_v$ according to the graph $G=(V,E)$.

\begin{algorithm}[H]
\caption{ The swallowing algorithm}\label{algo:swa}
\begin{algorithmic}[1]
\State Given an ordering of vertex $v_1,\ldots ,v_n$ of $G$.
\State Set $i\leftarrow 1$.  Set the current tensor $A[i]$ to be the tensor $M^{[v_1]}$, which can be pictured as one vertex and some free edges as in Figure \ref{fig:tensor3} (a).
\While{$i\leq n$}
\Comment{Adding tensor $M^{[v_{i+1}]}$ to $A[i]$}
\If {$M^{[v_{i+1}]}$ and $A[i]$ share edges in $G$} 
\State Construct a new tensor $A[i+1]$ by contraction, i.e. identifying the shared edges and summing over the corresponding indices. \label{l:add_M}
\Else
\State{$M^{[v_{i+1}]}$ and $A[i]$ has no common edge. 
\State Then  $A[i+1]$ is defined as the product  $M^{[v_{i+1}]} \otimes A[i]$. }
\EndIf
\State  $i\leftarrow i+1$
\EndWhile
\State\Return $\chi(T)\leftarrow A[n]$
\end{algorithmic}
\end{algorithm}

\textbf{High level ideas for Arad and Landau's result and our result.}
Let us first describe Arad and Landau's result at a high level before writing down formal statements with heavy notations. 
As in Figure \ref{fig:contraction}, given an arbitrary ordering to the vertices, for every vertex $v_i$, we implicitly partition the free edges of $M^{[v_i]}$ into input and output edges. 
With respect to this partition of in-edges and out-edges, one can write $M^{[v_i]}$ as a matrix denoted as $M^{[v_i]in,out}$.
As in Algorithm \ref{algo:swa},
the contraction value of the tensor network is then given by
sequentially mapping the in-edges  to out-edges, which can be represented by  
the matrix multiplication  $\prod_i M^{[v_i]in,out}\otimes I_{else}$, where  $I_{else}$ denotes the free edges other than the input edges in $A[i-1]$.   
 Arad and Landau's result shows that this matrix multiplication can be simulated by a quantum circuit through embedding each matrix $M^{[v_i]in,out}$ into a unitary, where the embedding is done by adding an ancillary qubit. Our result is, when every $ M^{[v_i]in,out}$ is a positive matrix, instead of embedding it into a unitary, we  embed the positive matrix into a stochastic matrix and simulate positive matrix multiplication with a random walk.

To explain Arad and Landau's result formally, we define more notation which is used in the swallowing algorithm. This notation is adapted from \cite{arad2010quantum}.
\begin{figure}[H]
    \centering
\includegraphics[width=0.4\textwidth]{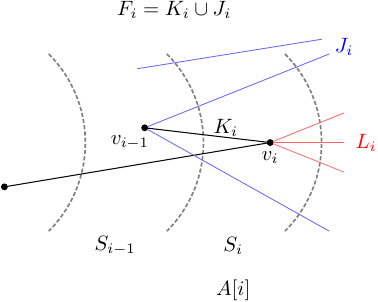}
% maybe add figure for algorithm later
%\includegraphics[width=\textwidth]{pics/swallowing.jpeg}
    \caption{ $F_i$ is the edges connecting $\{v_1,...,v_{i-1}\}$ and $\{v_i,...,v_n\}$. 
    The edges attached to $v_i$ are partitioned into the in-edges $K_i$ and out-edges $L_i$. When contracting the tensor $M^{[v_i]}$ we map the in-edges to out-edges. The edges in $F_i$ but not in $K_i$ are called $J_i$.   }
    \label{fig:contraction}
\end{figure}

  As in Figure \ref{fig:contraction}, define
\begin{itemize}
	\item  $S_i=\{v_1,\ldots ,v_{i}\}$.
    \item $F_i$ is the set of edges which connect $S_{i-1}$ and $V/S_{i-1}$.  $F_i$ are  the free edges in tensor $A[i-1]$. 
    \item  $K_i$ is the set of edges which connects $S_{i-1}$ and $v_{i}$. $K_i$  are the edges being contracted when contracting $A[i-1]$ and $M^{[v_{i}]}$. Note that $K_1=\emptyset$.
    \item $J_i\coloneqq  F_i/K_i$, which are the free edges in both $A[i-1]$ and $A[i]$. $J_1=\emptyset$.
    \item $L_i$ is the set of  edges of $v_{i}$ which are not in $K_i$. In other words, $L_i$ are the new free edges introduced by adding tensor  $M^{[v_{i}]}$ to $A[i-1]$.
\end{itemize}

Denote edges in $K_i$ as $\{e^{K_i}_1,\ldots ,e^{K_i}_{|K_i|}\}$. Denote edges in $J_i$ and $L_i$ similarly. With some abuse of notations, we use $e^{K_i}_s$ to denote both the name of the edge and the colors in $[d]$ that the edge $e^{K_i}_s$ takes. 

In the following, we explain that the update from tensors $A[i-1]$ to $A[i]$ can be written as matrix multiplication. More specifically:
\begin{itemize}
    \item First note that
 $A[i]$ can be viewed as a column vector consisting of $d^{|F_{i+1}|}$ entries, where the entries are indexed by the free edges of $A[i]$, that is
  $$F_{i+1}=K_{i+1}\cup J_{i+1} = L_i \cup J_i.$$ 
 More specifically,  write this vector as $\left| A[i]\right\rangle\in \bC^{d^{|F_{i+1}|}}$, then
  for edges in $J_i, L_i$ taking colors as  $e^{J_i}_1,\ldots ,e^{J_i}_{|J_i|}, e^{L_i}_1,\ldots ,e^{L_i}_{|L_i|}$ where  $e^{J_i}_s,e^{L_i}_t\in [d]$, we define
\begin{align}
	 \left\langle e^{J_i}_1,\ldots ,e^{J_i}_{|J_i|}, e^{L_i}_1,\ldots ,e^{L_i}_{|L_i|} \bigg| A[i] \right\rangle \coloneqq  A[i]_{e^{J_i}_1,\ldots ,e^{J_i}_{|J_i|}, e^{L_i}_1,\ldots ,e^{L_i}_{|L_i|} }.
\end{align}
\item  Note that $M^{[v_i]}$ can be viewed as a matrix:  When adding $M^{[v_i]}$ to $A[i-1]$ in Line \ref{l:add_M} of Algorithm \ref{algo:swa},  we contract the free edges in $K_i$ and introducing new free edges $L_i$.  One can view $M^{[v_i]}$  as a mapping from $\bC^{d^{|K_i|}}$ to $\bC^{d^{|L_i|}}$, denoted as $M^{[v_i]K_iL_i}$, which can be written as a matrix of size $d^{|L_i|}\times d^{|K_i|}$ where
\begin{align}
    \left\langle  e^{L_i}_1,\ldots ,e^{L_i}_{|L_i|} \bigg| M^{[v_i]} \bigg| e^{K_i}_1,\ldots ,e^{K_i}_{|K_i|}\right\rangle  = M^{[v_i]}_{e^{L_i}_1,\ldots ,e^{L_i}_{|L_i|},e^{K_i}_1,\ldots ,e^{K_i}_{|K_i|}}
\end{align}

In particular, since $K_1=\emptyset$, $M^{[v_i]K_iL_i}$ is a column vector in $\bC^{d^{|L_1|}}$.  
 For convenience,
   define $\ket{A[0]}$ to be a scalar,
  $$\ket{A[0]}\coloneqq 1.$$
  Denote $I_{J_i}$ as the identity operator on the indices with respect to edges in $J_i$.
  \item One can check that the updates from $A[i-1]$ to $A[i]$ can be written as matrix multiplication, that is 
  for $i=1,\ldots ,n,$
 \begin{align}\label{eq:update}
 	\left| A[i]\right\rangle = I_{J_i}\otimes M^{[v_i]K_iL_i} \ket{A[i-1]},
 \end{align}
 in the sense that
 \begin{align}
	&\left\langle e^{J_i}_1,\ldots ,e^{J_i}_{|J_i|}, e^{L_i}_1,\ldots ,e^{L_i}_{|L_i|} \bigg| A[i] \right\rangle \nonumber \\
	&= \sum_{ e^{K_i}_1,\ldots ,e^{K_i}_{|K_i|} } \left\langle  e^{L_i}_1,\ldots ,e^{L_i}_{|L_i|} \bigg| M^{[v_i]} \bigg| e^{K_i}_1,\ldots ,e^{K_i}_{|K_i|}\right\rangle  
	\left\langle e^{K_i}_1,\ldots ,e^{K_i}_{|K_i|},e^{J_i}_1,\ldots ,e^{J_i}_{|J_i|}\bigg|  A[i-1]\right\rangle.
\end{align}
In particular, $\ket{A[n]}$ is a scalar which equals to the contraction value. Thus we have
\begin{align}
    &\ket{A[n]}=\prod_{i=1}^n I_{J_i}\otimes M^{[v_i]K_iL_i}=\chi(T).\label{eq:chi_A}
\end{align}
\end{itemize}

To ease notations we define the \textit{swallowing operator} as 
\begin{align}\label{eq:swa}
O^{[v_i]} \coloneqq  I_{J_i}\otimes M^{[v_i]K_iL_i}.	
\end{align}

\subsection{$\BPP$-$\Completeness$ of  additive-error approximation}

According to the discussion in the previous section, one can compute $\chi(T)$ exactly by updating $\left| A[i]\right\rangle$ according to Eq.~\eqref{eq:update}.  It is well known that computing $\chi(T)$ exactly is $\#\PT$-$\hard$ even for a constant degree graph $G$, thus one cannot efficiently perform the exact version of the update Eq.~\eqref{eq:update}. 
However, interestingly 
\cite{arad2010quantum} showed that one can approximately perform the update 
efficiently using a quantum computer, where the approximation refers to an inverse polynomial additive error in the 2-norm of the tensors.

\begin{theorem}[Additive $2$-norm approximation of tensor networks~\cite{arad2010quantum}]\label{thm:BQP}
 Let $G=(V,E)$ be an $n$-vertex graph of constant degree. Let $T(G,M)$ be a tensor network on $G$ with bond dimension $d=O(poly(n))$. The following approximation problem is $\BQP$-$\Complete$: Given as input
 \begin{itemize}
     \item a tensor network $T(G,M)$, and a precision parameter $\epsilon=1/poly(n)$, and
     \item an ordering of the vertices $v_1,\ldots ,v_n$, and the corresponding swallowing operators $O^{[v_i]}$ defined in Eq.~\eqref{eq:swa},
 \end{itemize}
     output a complex number $\hat{\chi}(T)$ such that
  \begin{align}
	Pr\left(|\chi(T)-\hat{\chi}(T)|\leq \epsilon \Delta_2 \right) \geq 3/4,
\end{align}
	where
	\begin{align}
		\Delta_2\coloneqq  \prod_{i=1}^n \|O^{[v_i]}\|_2.
	\end{align}
\end{theorem}
Note that by Eq.~(\ref{eq:swa}), both the 2-norm $\|\cdot\|_2$ and 1-norm of $\|\cdot\|_1$ of $O^{[v_i]}$ equal to the corresponding norm of $M^{[v_i]K_iL_i}$.

We prove that  if the tensor network $T$ is a positive tensor network, then instead of using a quantum computer,   we can approximate the update Eq.~\eqref{eq:update} efficiently using a classical computer, where the approximation refers to an inverse polynomial  additive error in the matrix $1$-norm of the matrices $O^{[v_i]}$.
\begin{theorem}[Additive $1$-norm approximation of positive tensor networks]\label{thm:BPP}
 Let $G=(V,E)$ be an $n$-vertex graph of constant degree. Let $T(G,M)$ be a positive tensor network on $G$ with bond dimension $d=O(poly(n))$. The following approximation problem is $\BPP$-$\Complete$: Given as input
 \begin{itemize}
     \item a positive tensor network $T(G,M)$, and a precision parameter $\epsilon=1/poly(n)$, and
     \item an ordering of vertices $v_1,\ldots ,v_n$, and the corresponding swallowing operators $O^{[v_i]}$  defined in Eq.~\eqref{eq:swa},
 \end{itemize}
 output a complex number $\hat{\chi}(T)$ such that
  \begin{align}
	Pr\left(|\chi(T)-\hat{\chi}(T)|\leq \epsilon \Delta_1 \right) \geq 3/4,\label{eq:BPP}
\end{align}
	where
	\begin{align}
		\Delta_1\coloneqq  \prod_{i=1}^n \|O^{[v_i]}\|_1.	\label{eq:scale}
	\end{align}
\end{theorem}

\subsection{Proof of Theorem \ref{thm:BPP}}
The proof of the $\BPP$-$\hardness$ part for Theorem \ref{thm:BPP}
is similar to Section 4.2 in \cite{arad2010quantum}. For completeness we  give a proof sketch in Appendix \ref{appendix:BPPhard}. 
In the following we prove the ``inside $\BPP$'' part of  Theorem \ref{thm:BPP}, that is, we provide an efficient classical algorithm that achieves 
Eq.~\eqref{eq:BPP}. 
The main idea of the algorithm is to simulate non-negative matrix multiplication via a stochastic process.

We say that a matrix is \textit{non-negative} if all of its entries are non-negative. 
We first give a lemma which extends a non-negative matrix to a stochastic matrix. 

\begin{lemma}\label{lem:ext} 
 Let $M\neq 0$ be a non-negative matrix, that is  $M \in \bR_{\geq 0}^{m\times n}$. Then there exists
 $N\in \bR_{\geq 0}^{2m\times n}$ such that $N$ is a stochastic matrix\footnote{$N$ is a stochastic matrix iff $N$ is a non-negative matrix, and each column sums to $1$.}, and $\frac{M}{\|M\|_1}$ is the first $m$ rows of~$N$.  
\end{lemma}

\begin{proof}
	Since  
 $\|M\|_1 = \max_{1\leq j\leq n} \sum_{i=1}^m |M_{ij}|,$ thus for each column, the  column sum 
	 of $\frac{M}{\|M\|_1}$ lies in $[0,1]$, thus one can  embed $\frac{M}{\|M\|_1}$ as the first $m$ rows in a stochastic matrix~$N$.
\end{proof}

Before we state the $\BPP$ algorithm, we recall some notations. 
Recall that we have defined  $F_i,J_i,K_i,L_i, A[i],M^{[v_i]K_iL_i}$  in Section \ref{sec:SWA} and Figure \ref{fig:contraction}. To ease notation, we abbreviate 
$M^{[v_i]K_iL_i}$ as $M^{[v_i]}$.  
Since we are working with positive tensor network, 
 $M^{[v_i]}$  is a non-negative matrix and the entries are indexed by $K_i,L_i$.
  
  Lemma \ref{lem:ext} says that we can embed $M^{[v_i]}/\|M^{[v_i]}\|_1$ in a stochastic matrix $N^{[v_i]}$ by adding one ancillary bit, that is $N^{[v_i]}$ is indexed by $L_i\cup\{w_i\}$ and $K_i$, where $w_i\in\{0,1\}$ is an index such that $w_i=0$ (or 1) refers to the first (or second) $d^{|L_i|}$ rows of $N^{[v_i]}$. Define 
 \begin{align}
 	W_i \coloneqq  \{w_1,\ldots ,w_{i-1}\}.
 \end{align}

We first explain the high level idea of the $\BPP$ algorithm, and then give the pseudo code.
The idea of the $\BPP$ algorithm in Theorem \ref{thm:BPP}  is to mimic  non-negative matrix multiplication by stochastic methods. 
From a high level idea, we will embed the vector $|A[i]\rangle$ in a probability distribution, and embed the matrix $M^{[v_i]}$ in a stochastic matrix  $N^{[v_i]}$ by adding an ancillary bit. Then the update rule
	\begin{align}
 		\left| A[i+1]\right\rangle = I_{J_i}\otimes M^{[v_i]K_iL_i} \ket{A[i]},
 	\end{align}
 	is embedded in applying the stochastic matrix $N^{[v_i]}$ to the distribution $A[i]$, which can be simulated by a  random walk.

 Before writing down the pseudo codes we explain some notations.
 \begin{itemize}
     \item We use $\ket{k}_K$ to represent the (ordered) coloring $k\in[d]^{ |K|}$ of the edges in $K$. 
     When $K=\emptyset$ we view $\ket{k}_K$ as empty, that is writing down nothing.  Similarly for $\ket{j}_J,\ket{w}_W$.
     \item We will use $s_i$ to denote a computational basis whose distribution embed the
 vector $A[i-1]$.
 \item Recall that the free edges of $A[i-1]$ are $F_i=K_i\cup J_i$.  Also recall that $F_{i+1} = J_i\cup L_i$. 
 \end{itemize}

  The pseudo codes are stated as  Algorithm \ref{alg:trial} and Algorithm \ref{alg:BPP}.  Their performance is given in Corollary \ref{cor:BPP}. 
	 
% Potentially add figure later

%\includegraphics[width=\textwidth]{pics/swallowing_markov.jpeg}}

\begin{algorithm}
\caption{Trial(T=T(G,M))}\label{alg:trial}
\begin{algorithmic}[1]
\State $i=1$, $s_1=\emptyset$.
\While{$i\leq n$}
	\State Interpret $s_{i} = \ket{k}_{K_i}\ket{j}_{J_{i}}\ket{w}_{W_{i}}$.
	\Comment{$ K_1=  J_1=W_1=\emptyset.$}
	\State Recall that $N^{[v_i]}$ has rows indexed by $L_i\cup \{w_{i}\}$, column indexed by $K_i$.
	\State Since $N^{[v_i]}$ is a stochastic matrix, each column of  $N^{[v_i]}$ corresponds to a distribution over the row index.
 \If{i=1}
 \State $N^{[v_1]}$ is a column vector. Denote the distribution according to  this  column as $\cD$.
 \Else{}
 \State Denote the distribution according to  the $k$-th column of $N^{[v_i]}$ as $\cD$.
 \EndIf
	\State Sample a row 
	index according to $\cD$. \label{line:11} \State Denote this row index as $\ket{lw'}_{L_i\cup \{w_{i}\}}$, where $l \in  [d]^{|L_i|},w'\in\{0,1\}.$
	\State Set $s_{i+1}\leftarrow \ket{l}_{L_i}\ket{j}_{J_{i}}\ket{ww'}_{W_{i+1}}$.
	\Comment{ $N^{[v_i]}$ maps register $K_i$ to $L_i\cup \{w_i\}$.}
	\State $i\leftarrow i+1$.
\EndWhile
\State Interpret  $s_{n+1} = \ket{w}_{W_{n+1}}$
\Comment{We know $L_{n}=J_{n}=\emptyset$}
\State Return Success if $w=00\ldots 00$, that is the all zero state; otherwise return Fail.
\end{algorithmic}
\end{algorithm}

\begin{algorithm}
\caption{Approximating positive tensor network}\label{alg:BPP}
\begin{algorithmic}[1]
\State Set $K=10 \epsilon^{-2}$
\State Run the  $Trial(T)$, that is Algorithm \ref{alg:trial}, for $K$ times. 
\State Count the number of Success as $\# Success$.
\State Return $\hat{\chi}(T)\leftarrow \frac{\# Success}{K} \cdot \Delta_1$ as the approximation of $\chi(T)$.
\end{algorithmic}
\end{algorithm}

\begin{lemma}\label{lem:trial}
	The probability of the Trial algorithm (Algorithm \ref{alg:trial}) to return ``Success'' is $\chi(T)/\Delta_1$ where  $\Delta_1\coloneqq  \prod_{i=1}^n \|O^{[v_i]}\|_1$. Its runtime is $poly(n)$.
\end{lemma}
\begin{proof}
	To ease notation, for a set $Q\cup W$ where $Q$ is a set of edges of $G$, and $W$ is a set of ancillary indices $\{\ldots ,w_i,\ldots \}$, we use 
 $$String(Q\cup W) \coloneqq   [d]^{ |Q|} \times \{0,1\}^{ |W|} $$ where edges/indices in $Q$ take values in $[d]$, and ancillary indices in $W$ takes values in $\{0,1\}$.
	
	Define 
 $$S_i\coloneqq  String ( L_i \cup J_{i} \cup W_{i+1}),$$ define $\ket{\pi_{i+1}}$ as the probability distribution of $s_{i+1}$, which is a probability distribution over $S_i$. Note that
	\begin{align}
 \ket{\pi_2} &= I_{J_1}\otimes N^{[v_1]}, \text{ where }I_{J_1}=I_{\emptyset}=1 \\ 
			\ket{\pi_{i+1}} &= I_{J_i}\otimes N^{[v_i]} \ket{\pi_i} \nonumber\\
 &=  \prod_{h=1}^i I_{J_h}\otimes N^{[v_h]} \label{eq:update_BPP} \\
			 &=\prod_{h=1}^i \frac{I_{J_h}\otimes M^{[v_h]}}{\|M^{[v_h]}\|_1} \ket{0\ldots 0}_{W_{i+1}}  + \text{additional terms}   \label{eq:NV} \\
			 &=\prod_{h=1}^i \frac{O^{[v_h]}}{\|O^{[v_h]}\|_1} \ket{0\ldots 0}_{W_{i+1}}  + \text{additional terms}
	\end{align}
 where Eq.~\eqref{eq:update_BPP} is from Algorithm \ref{alg:trial}; Eq.~(\ref{eq:NV}) is from the definition of $N^{[v_h]}$ that $N^{[v_h]}$ embeds $M^{[v_h]}$; and
	 the last equality comes from definition of $O^{[v_h]}$,
	\begin{align}
		&O^{[v_h]} = I_{J_h}\otimes M^{[v_h]},\\
		&\|O^{[v_h]}\|_1=\|M^{[v]}\|_1.
	\end{align}
From Eq.~(\ref{eq:chi_A})	we conclude
	\begin{align}
	\label{eq:final outcome sampling}
			\ket{\pi_{n+1}}
			 &=\frac{\chi(T)}{ \Delta_1} \ket{0\ldots 0}_{W_{n+1}}  + \text{additional terms} 
	\end{align}
To prove the runtime is $poly(n)$, notice that in Theorem \ref{thm:BPP} we assume that $G$ is a graph of constant degree, thus $|L_i|,|K_i|$ are constants, thus $d^{|L_i|},d^{|K_i|}$ are $poly(n)$ whenever $d = poly(n)$  and 
$N^{[v_i]}$ is a matrix of size $poly(n)\times poly(n)$. Thus Line \ref{line:11} in Algorithm \ref{alg:trial} can be done efficiently.
\end{proof}

\begin{corollary}\label{cor:BPP}
	For $t=poly(n)$ sufficiently large,  the output $\hat{\chi}(T)$ in Algorithm \ref{alg:BPP} satisfies
		  \begin{align}
	&Pr\left(|\chi(T)-\hat{\chi}(T)|\leq \epsilon \Delta_1 \right) \geq 3/4
	\end{align}
\end{corollary}
where $\Delta_1\coloneqq  \prod_{i=1}^n \|O^{[v_i]}\|_1.$

\begin{proof}
The proof directly follows from Eq.~\eqref{eq:final outcome sampling} and Chebyshev's inequality.
Write $X_i$ as the result of the $i$-th trial, where $X_i=1$ if the Trial algorithm (Algorithm~\ref{alg:trial}) returns Success and $=0$ otherwise. By Lemma \ref{lem:trial}, we have 
$$E [X_i]=\chi(T)/\Delta_1.$$
Besides, note that $|\chi(T)/\Delta_1|\leq 1$ since it corresponds to a probability, we have  
\begin{align}
	|Var(X_i)| &\leq E|X_i|^2 +| E [X_i] |^2\leq E|X_i|^2+ |\chi(T)/\Delta_1|^2\leq 2.
\end{align}

Define $X\coloneqq (X_1 +\ldots +X_K)/K$, we have 
\begin{align}
	&E(X)= 	\frac{\chi(T)}{\Delta_1},\\ 
	&Var(X) =\frac{1}{K}Var(X_1) \leq \frac{2}{K}
 \end{align}
 By definition of $\hat{\chi}(T)$ in Algorithm \ref{alg:BPP}, 
use Chebyshev's inequality we have
\begin{align}
	Pr\left(|\chi(T)-\hat{\chi}(T)|\geq \epsilon \Delta_1 \right) & = Pr\left(|\frac{\chi(T)}{\Delta_1}- X|\geq \epsilon \right)\\
	& \leq \frac{Var(X)}{\epsilon^2}\\
	& \leq  \frac{2}{K\epsilon^2}\\
	&= 1/5.
\end{align}
\end{proof}

\subsection{$\StoqMA$-$\hardness$ of multiplicative-error approximation}\label{sec:StoqMA}

In this section, we consider the task of approximating a  positive tensor network up to a multiplicative error. We show that this approximation is $\StoqMA$-$\hard$ up to exponentially close to 100\% error.

\begin{theorem}\label{thm:StoqMA}
Let $G=(V,E)$ be an $n$-vertex graph of constant degree. Let $T(G,M)$ be a positive tensor network on $G$ with bond dimension $d=poly(n)$.  Consider the following approximation problem: Given as inputs a tensor network $T(G,M)$ and a precision parameter $\epsilon \leq 1- \exp(-n)$, output a complex number $\hat{\chi}(T)$ such that 
	\begin{align}
		Pr\left( \left|\chi(T) - \hat{\chi}(T)\right| \leq \epsilon |\chi(T)|  \right)\geq 3/4.	
	\end{align}
	If there exists a $poly(n)$-time  randomized algorithm for solving the above approximation problem, then there exists a  $poly(n)$-time randomized  algorithm for solving $\StoqMA$ with probability greater than $3/4$.
\end{theorem}
Note that $\epsilon\leq 1-\exp(-n)$ means we allow very large (close to 100\%) multiplicative error.
Recall that $\StoqMA$ is a subclass of $\QMA$  which is related to  deciding ground energy  for stoquastic Hamiltonians~\cite{bravyi2006merlin}. 
For our purpose, we use an equivalent definition of $\StoqMA$ that makes use of the notion of a stoquastic verifier.

\begin{definition}[$\StoqMA$,  from \cite{bravyi2006merlin}]\label{def:StoqMA}
	A stoquastic verifier is a tuple $V=(n,n_w,n_0,n_+,U)$, where
 \begin{itemize}
     \item $n$ is the number of input bits, $n_w$ is the number of input witness qubits.
     \item  $n_0$ is the number of input ancillas $\ket{0}$, $n_+$ is the number of input ancillas $\ket{+}$.
     \item $U$ is a quantum circuit on $n+n_w+n_0+n_+$ qubits with $\rm{X}$, $\rm{CNOT}$, and $\rm{Toffoli}$ gates. 
 \end{itemize}
   The acceptance probability of a stoquastic verifier $V$ on input string $x\in\Sigma^n$ and witness state $\ket{\psi}\in (\bC^2)^{n_w}$ is defined as 
	\begin{align}
		& Pr(V;x,\psi) = \langle \psi_{in} |U^\dagger \Pi_{out} U |\psi_{in}\rangle,\\
	\text{where }	& \ket{\psi_{in}} = \ket{x}\otimes \ket{\psi} \otimes \ket{0}^{\otimes n_0} \otimes \ket{+}^{\otimes n_+},\\
		&\Pi_{out} = \ket{+}\bra{+}_1\otimes I_{else}. 
	\end{align}
	A promise problem $L=L_{yes}\cup L_{no}\subseteq \Sigma^*$ belongs to $\StoqMA$ iff there exists a uniform family of stoquastic verifier $V$ which uses at most $poly(n)$ qubits and gates, and obeys the following:
\begin{itemize}
	\item Completeness. If $x\in L_{yes}$, then there exists $\ket{\psi}$ such that $Pr(V;x,\psi)\geq b$;
	\item Soundness. If $x\in L_{no}$, then for any $\ket{\psi}$ we have $Pr(V;x,\psi)\leq a$;
\end{itemize}
where $0\leq a < b \leq 1$ and $b-a\geq 1/poly(n)$.
\end{definition}

\begin{figure}[H]
    \centering
\includegraphics[width=1\textwidth]{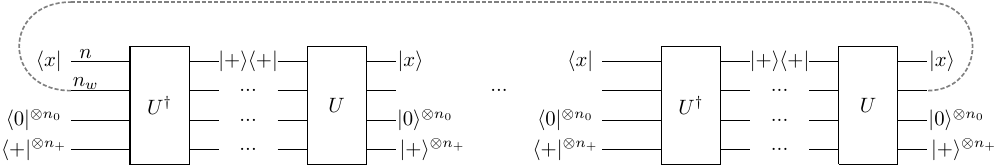}
    \caption{The above figure contains  two copies of $M_x$. When connecting the right side of the first $M_x$ and the left side of the second $M_x$, we get the
 operator $M_x^2$. Further connecting the left side of the first $M_x$  and the right side of the second $M_x$ by the dashed line, we get $tr(M_x^2)$.}
    \label{fig:StoqMA}
\end{figure}

The proof of Theorem \ref{thm:StoqMA} is adapted from the folklore proof of $\QMA\subseteq \PP$, where the adaption is mainly translating matrix operations (multiplication, trace, etc) to  tensor network operations. We only give a proof sketch here.

\begin{proof}[Proof of Theorem \ref{thm:StoqMA}]
In this proof we use the notions in Definition \ref{def:StoqMA}. 
Consider a language $L=L_{yes}\cup L_{no}$ in $\StoqMA$ with stoquastic verifier $V$ in Definition \ref{def:StoqMA}. For input $x$, as pictured in Figure \ref{fig:StoqMA} we
define a positive semi-definite Hermitian operator acting on $n_w$ qubits
as
\begin{align}
		&M_x \coloneqq  \langle \phi |U^\dagger\Pi_{out} U  |\phi\rangle,\\
	\text{where }	&|\phi\rangle = \ket{x} \otimes \ket{0}^{\otimes n_0}\otimes \ket{+}^{\otimes n_+}.
\end{align}
Denote the maximum eigenvalue of $M_x$ as $\lambda_{\max}(x)$.
By assumption we have that 
\begin{itemize}
	\item If $x\in L_{yes}$, then $\lambda_{\max}(x)\geq b.$
	\item If $x\in L_{No}$, then $\lambda_{\max}(x)\leq a.$
\end{itemize}
Since $M_x$ is positive semi-definite, then for any $k\in \bN$ we have
\begin{itemize}
	\item If $x\in L_{yes}$, then $tr(M_x^k) \geq b^k$.
	\item If $x\in L_{No}$, then $tr(M_x^k) \leq  2^{n_w} a^k$.
\end{itemize}
Recall that $\epsilon\leq 1- \exp(-n)$ is the precision parameter. By the assumption in Theorem \ref{thm:StoqMA}, there is a $poly(n)$-time randomized algorithm $\cA$ such that  with  probability at least $3/4$, the algorithm returns
\begin{align}
	(1-\epsilon) |\chi(T)| \leq |\hat{\chi}(T)| \leq (1+\epsilon) 	|\chi(T)|.
\end{align}
To  distinguish the yes and no cases,
 set 
 $$k\geq \left( n_w\ln 2 +  \ln \frac{1+\epsilon}{1-\epsilon} \right) / \ln \frac{b}{a} =poly(n),$$
 We have
\begin{align}
		2^{n_w}a^k (1+\epsilon) < b^k (1-\epsilon).
\end{align}
Thus  one can distinguish whether $x\in L_{yes}$ or $x\in L_{no}$ by approximating $tr(M_x^k)$  using the algorithm $\cA$.

It remains to explain  that $tr(M_x^k)$ can be represented by a positive tensor network $T=T(G,M)$ with $poly(n)$ bond dimension,  where $G$ is a  $poly(n)$-vertex graph of constant degree. 

First notice that 
similarly as Section \ref{appendix:BPPhard} or Section 4.2 in \cite{arad2010quantum},  one can naturally represent $M_x$ as a  tensor network $T=T(G,M)$ with $poly(n)$ bond dimension. Since the gates in $U$  are $\rm{X}$, $\rm{CNOT}$, and $\rm{Toffoli}$, and the ancillas are computational basis or $\ket{+}$, one can check that $T(G,M)$ is a positive tensor network. Further, since $U$ has $poly(n)$ gates and each gate  has constant number of input qubits and output qubits, we have that $G$ is a $poly(n)$-vertex graph of constant degree.  

To represent $tr(M_x^k)$ as a tensor network, as in Figure \ref{fig:StoqMA},
it suffices to additionally notice that
\begin{itemize}
\item The tensor network for the operator $M_x^2$ can be represented by putting $2$ copies of $M_x$ in a line, then connecting the right side of the  first $M_x$ and the left side of the second $M_x$, that is 
 contracting the free edges w.r.t register $n_w$ for the first and second copy. 
 \item Similarly as explained in Eq.~(\ref{eq:trace}) in Section \ref{sec:Notation}, if we further connect the left side of the first $M_x$  and the right side of the second $M_x$ by the dash line, we get $tr(M_x^2)$.
	\item The tensor network for the operator $tr(M_x^k)$ can be represented similarly. That is  putting $k$ copies of $M_x$ in a line, and then contracting the free edges w.r.t register $n_w$ sequentially.
\end{itemize}
\end{proof}

\section{Acknowledgements}
We thank Garnet Chan and  Zeph Landau for helpful discussions.
Part of this work was conducted while the authors were visiting the Simons Institute for the Theory of Computing during summer 2023 and spring 2024, supported by DOE QSA grant $\#$FP00010905.
%whose hospitality we gratefully acknowledge.
D.H.\ acknowledges financial support from the US DoD through a QuICS Hartree fellowship.
N.S.\  acknowledges financial support by the Austrian Science Fund FWF (Grant DOIs \href{https://doi.org/10.55776/COE1}{10.55776/COE1}  and \href{https://doi.org/10.55776/F71}{10.55776/F71}) and the European Union’s Horizon 2020 research and innovation programme through Grant No.\ 863476 (ERC-CoG SEQUAM).
Jiaqing Jiang is supported by MURI Grant FA9550-18-1-0161 and the IQIM, an     NSF Physics Frontiers Center (NSF Grant PHY-1125565). Jielun Chen is supported by the US National Science Foundation under grant CHE-2102505.

\appendix

\section{$\#\PT$-$\hardness$ of exactly contracting random 2D tensor networks}\label{appendix:exact-case}

Here we prove that the exact contraction of  the random 2D tensor network with a positive mean remains $\#\PT$-$\hard$. 

Firstly we prove some properties of standard Gaussian distribution, while the finite precision Gaussian distribution behaves similarly up to $O(\exp(-n))$ derivation in the error bounds.
Recall that we use $X\sim \cN_{\bC}(\mu,\sigma^2)$ (or $X\sim \cN_{\bR}(\mu,\sigma^2)$) to denote that the random variable $X$ is sampled from  the complex (or real) Gaussian distribution with mean $\mu$ and standard derivation $\sigma$. We use $\vec{X}=(X_1,X_2,\ldots ,X_m)\sim \prod_{i=1}^m \cN_{\bC}(\mu_i,\sigma^2)$ to denote the random variable $\vec{X}$ where each $X_i$ is independently sampled from $\cN_{\bC}(\mu_i,\sigma^2)$. When $\mu_i =\mu, \forall i$, we abbreviate the notation as $\vec{X}\sim \cN_{\bC}(\mu,\sigma^2)^m$.
For two distribution $\cD_1,\cD_2$, we use $\|\cD_1-\cD_2\|$ to denote the total variation distance.

\begin{lemma}[Analogy of Lemma 5 in \cite{haferkamp2020contracting}]\footnote{There is a remark on the notation difference.  \cite{haferkamp2020contracting}  uses $\cN_\bC(\mu,\sigma)$ to denote Gaussian distribution with mean value $\mu$ and standard derivation $\sigma$. In this manuscript we denote this distribution as $\cN_{\bC}(\mu,\sigma^2)$ which is the more standard notation.} \label{lem:error}   
For $\mu_i\in\bC$. It holds that
\begin{align}
    &\|\cN_{\bC}(\mu,(1-\epsilon)^2\sigma^2)^m - \cN_{\bC}(\mu,\sigma^2)^m \| \leq 4m\epsilon\\
    &\|\prod_{i=1}^m \cN_{\bC}(\mu_i,\sigma^2)- \cN_{\bC}(\mu,\sigma^2)^m\| \leq \frac{2}{\sigma} \left(|\mu_1-\mu|+\ldots +|\mu_m-\mu|\right) 
\end{align}
\end{lemma}
\begin{proof}
Recall that
for $\mu\in \bC$, we use $\Re(\mu),\Im(\mu)\in \bR$ for the real and imaginary part of $\mu$, that is $\mu= \Re(\mu) +\Im(\mu) i$. Besides, $X\sim \cN_{\bC}(\mu,\sigma^2)$ iff $\Re(X)\sim \cN_{\bR}(\Re(\mu),\frac{\sigma^2}{2}), \Im(X)\sim \cN_{\bR}(\Im(\mu),\frac{\sigma^2}{2})$.
 It suffices to notice that
    \begin{align}
          &\| \prod_{i=1}^m \cN_{\bC}(\mu_i,\sigma_i^2)^m - \cN_{\bC}(\mu,\sigma^2)^m\|\\ 
           & = \frac{1}{2}
\int_{(x_1,\ldots ,x_m)} \left| \prod_{i=1}^m  \frac{1}{ 2\pi (\sigma_i/\sqrt{2})^2} \exp^{-\frac{1}{2} |\frac{x_i-\mu_i}{\sigma_i/\sqrt{2}}|^2} -  
\prod_{i=1}^m  \frac{1}{2\pi (\sigma/2)^2} \exp^{-\frac{1}{2} |\frac{x_i-\mu}{\sigma/2} |^2} 
 \right|dx_1 \ldots  d x_m
\\
 & = \frac{1}{2}
\int_{(x_1,\ldots ,x_m)} \left| \prod_{i=1}^m  \frac{1}{2\pi (\sigma_i/\sqrt{2})^2} \exp^{-\frac{1}{2} |\frac{x_i-\mu_i+\mu}{\sigma_i/\sqrt{2}} |^2} -  
\prod_{i=1}^m  \frac{1}{2\pi (\sigma/2)^2} \exp^{-\frac{1}{2} |\frac{x_i}{\sigma/2}|^2} 
 \right|dx_1 \ldots  d x_m
\\
          &=\| \prod_{i=1}^m\cN_{\bC}(\mu_i-\mu,\sigma_i^2)^m - \cN_{\bC}(0,\sigma^2)^m\|
    \end{align}

   Thus by Lemma 5 in \cite{haferkamp2020contracting}, we have
    \begin{align}
        \| \cN_{\bC}(\mu,(1-\epsilon)^2\sigma^2)^m - \cN_{\bC}(\mu,\sigma^2)^m\|   &=   \| \cN_{\bC}(0,(1-\epsilon)^2\sigma^2)^m - \cN_{\bC}(0,\sigma^2)^m\| \\
       &\leq 2\times 2m\epsilon\\
        \|\prod_{i=1}^m \cN_{\bC}(\mu_i,\sigma^2)- \cN_{\bC}(\mu,\sigma^2)^m\| &= \|\prod_{i=1}^m \cN_{\bC}(\mu_i-\mu,\sigma^2)- \cN_{\bC}(0,\sigma^2)^m\|\\
        & \leq 2\times\frac{1}{\sigma}\left(|\mu_1-\mu|+\ldots +|\mu_m-\mu|\right)
    \end{align}
    where we add a $2\times$ since we are working with $\cN_{\bC}$ while Lemma 5 in \cite{haferkamp2020contracting} is with $\cN_{\bR}$.
     \end{proof}

\begin{proof}[Proof of Theorem \ref{thm:exactH}] The proof follows directly from  the  proof idea of Theorem 2 in \cite{haferkamp2020contracting}. Here we only give a proof sketch.

 Firstly 
\cite{haferkamp2020contracting,schuch2007computational}
showed that  
one can encode any $n$-variable boolean function  $f(x_1,\ldots ,x_n)$ into a projected entangled-pair states (PEPS) of $poly(n)$ vertices, which describes an un-normalized state $\ket{\psi}$, such that  computing $\langle \psi|\psi\rangle$ exactly is equivalent to  computing the value 
$$s(f)= |\{x\in\{0,1\}^n:f(x)=1\}|,$$
which  is $\#\PT$-$\Complete$. One can check that in this case $\langle \psi|\psi\rangle$ equals to the contraction value of a 2D\footnote{ $\langle \psi|\psi\rangle$ is a stack of two PEPS.  One can  transform it into a 2D tensor network by contracting the stack of two PEPS via free boundaries of the two PEPS.} tensor network of $poly(n)$ vertices, where  the 2D tensor network has bond dimension $d=O(poly(n))$, and 
every entry of the tensor network is bounded by a constant. 
one can further make the underlying 2D lattice for the 2D tensor network to  
have periodic boundary condition, by adding edges connecting boundaries and 
slightly modify the tensors near the boundary to make sure the contraction value remains invariant. 
Denote the final 2D lattice with periodic boundary condition
 as $G$. 
Denote the final  2D tensor network   which encodes the fixed boolean function $f$ as 
$$T\left(G,(P^{[v]})_v\right)$$ where $P^{[v]}$ is the tensor on vertex $v$. Note that $P^{[v]}$ has $d^4$ entries thus  $(P^{[v]})_v$ are described by in total $d^4 \times n$ entries.  For convinience, we assign an arbitrary order to those entries and denoted them as $\{p_i\}_{i=1}^{d^4 n}$. Recall that by construction we have $|p_i|\leq c$ for constant $c$.

 Theorem \ref{thm:exactH} is  proved by an argument of    average to worse case reduction  via interpolation.  Here we define the polynomial for the interpolation.
Set 
 \begin{align}
 	\epsilon &\coloneqq  \min\left\{\frac{1}{4 (c+\mu+1) d^4 n^3}\,,\,\frac{1}{2}\right\}.\\
 	k &= poly(n) \text{ be sufficiently large.} 	\\
 	\text{Let } &\text{$S=\{t_i\}_{i\in [k]}$ be the set of $k$ equidistant points in $[0,\epsilon]$.  }
 \end{align}
Recall that $0\leq u\leq poly(n)$ thus $\epsilon=1/poly(n)$.

We randomly sample a 2D $(\mu,n,d)$-Gaussian tensor network $T\left(G,(Q^{[v]})_v\right)$. 
 Let $t\in S$, define a new 2D tensor network $T(t)$, where for vertex $v$ the tensor $R(t)^{[v]}$ is defined as
 \begin{align}
 	 R(t)^{[v]}  \coloneqq  t  P^{[v]}  + (1-t)  Q^{[v]}.	\label{eq:159}
 \end{align}
Denote the exact contraction value of $T(t)$  as $q(t)$. Note that $q(t)$ is a degree-$n$ polynomial of $t$. Besides, from construction we know that computing $q(1)$ will solve $\#\PT$-$\Complete$ problem.

In the following, we show that if one can compute the exact contraction value of finite precision 2D Gaussian tensor network with high probability, then we can compute $q(1)$ with high probability by interpolation. More specifically,
For input $T(t)$, 
denote the value returned by the algorithm $\cA$ in Theorem \ref{thm:exactH} as $\cA(t)$. 

\textbf{($\romannumeral1$) First we prove that since $t$ is small, $\cA(t)$ is a good approximation of $q(t)$.} 
Specifically, define 
\begin{align}
	\mu_i= tp_i+ (1-t)u.
\end{align}
By Eq.~(\ref{eq:159}), we know that the entries of $\left( R(t)^{[v]} \right)_v$ are sampled from distribution
\begin{align}
	\cD \coloneqq  \prod_{i=1}^{d^4 n} \overline{\cN}_{\bC}(\mu_i, (1-t)^2).	
\end{align}
Since $\overline{\cN}_\bC$ approximates $\cN_\bC$ within exponential precision, we know that
\begin{align}
\|\cD - \overline{\cN}_{\bC}(\mu,1)^{d^4 n}\| & \leq O(\exp(-n)) + \|\prod_{i=1}^{d^4 n} \cN_{\bC}(\mu_i, (1-t)^2) - \cN_{\bC}(\mu,1)^{d^4 n}\| \label{eq:N_finite}
\end{align}
where by Lemma \ref{lem:error} we have
\begin{align}
	\|\prod_{i=1}^{d^4 n} \cN_{\bC}(\mu_i, (1-t)^2) - \cN_{\bC}(\mu,1)^{d^4 n}\| &\leq   \| \prod_{i=1}^{d^4 n} \cN_{\bC}(\mu_i, (1-t)^2) - \cN_{\bC}(\mu,(1-t)^2)^{d^4 n}  
	\| \nonumber\\
 &+ \|\cN_{\bC}(\mu,(1-t)^2)^{d^4 n} -  \cN_{\bC}(\mu,1)^{d^4 n}
	\| \nonumber\\
	&\leq  \frac{2}{(1-t)} \left( |\mu_1-\mu| + \ldots  |\mu_{d^4 n}-\mu|  \right) + 4 \cdot d^4 n \cdot t, \nonumber\\
		& \leq 4 \cdot d^4n \cdot (c+\mu) \epsilon + 4 \cdot d^4 n\cdot \epsilon,\label{eq:aa}\\
	&\leq 4\cdot d^4 n\cdot (c+\mu+1)\epsilon \nonumber\\
	& \leq \frac{1}{n^2} \label{eq:t_same}
\end{align}
where Eq.~(\ref{eq:aa}) comes from the facts that
\begin{align}
    &t\leq \epsilon \leq 1/2,\\
    &|\mu_i-\mu|= |t(p_i-\mu)|\leq (c+\mu)\epsilon.
\end{align}
 Eqs.~(\ref{eq:N_finite})(\ref{eq:t_same}) together imply 
 \begin{align}
     \|\cD - \overline{\cN}_{\bC}(\mu,1)^{d^4 n}\| & \leq O(\exp(-n)) + \frac{1}{n^2}.\label{eq:almost_same}
 \end{align}
 In other words, for any $t_i\in S$, the distribution of $\left( R(t)^{[v]} \right)_v $ is almost the same as the finite precision  2D $(\mu,n,d)$-Gaussian tensor network.  
 Let $n$ and $k=poly(n)$ be sufficiently large. By Eq.~(\ref{eq:almost_same})  and the assumption of the performance of $\cA$ we have
\begin{align}
 &Pr\left(
 \cA(t_i) = q(t_i)   \right)
 \geq \frac{3}{4} + \frac{1}{n} - O(\exp(-n))-\frac{1}{n^2}  \geq \frac{3}{4} + \frac{1}{n^2}.\\
 & E |\{i:\cA(t_i) = q(t_i)\}| \geq   \left(\frac{3}{4} + \frac{1}{n^2}\right)k
\end{align}
where in the second inequality $E$ refers to expectation. 
By Chernouff bound we know that  for sufficiently large  $k=poly(n)$, 
\begin{align}
Pr\left( |\{i:\cA(t_i) = q(t_i)\}|\geq \frac{k+n}{2} \right) &\geq 1 - \exp(-n). \label{eq:BW}
\end{align}

 \textbf{($\romannumeral2$)
 We then use the following theorem to recover the degree $n$ polynomial $q(t)$} 
 \begin{theorem}[Berlekamp-Welch~\cite{movassagh2018efficient}]\label{thm:BW}
 Let $q$ be a degree-$n$ polynomial over any field $\bF$. Suppose we are given $k$ pairs of elements $\{(x_i,y_i)\}_{i=1}^k$ with all $x_i$ distinct, and with the promise that $y_i=q(x_i)$ for at least $\max(n+1,(k+n)/2)$ points. Then, one can recover $q$ exactly in $poly(k,n)$ deterministic time.	
\end{theorem}

In Theorem \ref{thm:BW} let
\begin{align}
    &x_i\coloneqq  t_i,\\
    &y_i\coloneqq  \cA(t_i).
\end{align}
 Then by Eq.~(\ref{eq:BW}) we can recover $q(t)$ with probability $1 - \exp(-n)$ in $poly(n)$-time.

\textbf{($\romannumeral3)$ Finally}, we have $q(t)$ in hand, which is correct with probability $1 - \exp(-n)$. Since  $q(t)$ is a degree $n$ polynomial, we can easierly compute $q(1)$, which solves a $\#\PT$-$\Complete$ problem.

\end{proof}

\section{More on Barvinok's method}
\label{appendix:barvinok}

\begin{proof}[Proof of Lemma \ref{lem:root}]
The proof uses  Jensen's formula and follows the idea from \cite{eldar2018approximating}.
Let $z_1,\ldots ,z_j,\ldots $ be the roots of $h_A(z)$,  Jensen's formula establishes the connection between the roots in the disk $\cB(r)$, and the average of $\ln |h_A(z)|$ on the boundary of $\cB(r)$: 
\begin{align}
    \sum_{|z_j|\leq r} \ln \frac{r}{|z_j|} + \ln |h_A(0)| = E_\theta \ln |h_A(re^{i\theta})|. \label{eq:188}
\end{align}
First notice that
\begin{align}
     \sum_{|z_j|\leq r} \ln \frac{r}{|z_j|} &\geq \sum_{|z_j|\leq r(1-\nw)} \ln \frac{r}{|z_j|}\\
     &\geq \sum_{|z_j|\leq r(1-\nw)}  \ln \frac{r}{r(1-\nw)}\\
     &\geq \nw\cdot N_A (r(1-\nw)). \label{eq:191}
\end{align}
where in the last inequality we use $\ln \frac{1}{(1-\nw)}\geq \nw$ for $\nw\leq 1/2$.
Thus by Eqs.~(\ref{eq:191})(\ref{eq:188}), we have
\begin{align}
    E_A \left[N_A(r-r\nw)\right] &\leq \frac{1}{\nw} E_\theta E_A \ln \frac{ |h_A(re^{i\theta})|}{|h_A(0)|}\\
    &= \frac{1}{2\nw} E_\theta E_A \ln \frac{|h_A(re^{i\theta})|^2}{|h_A(0)|^2}\\
    &\leq \frac{1}{2\nw} \ln E_\theta E_A \frac{|h_A(re^{i\theta})|^2}{|h_A(0)|^2}.
\end{align}
where the last inequality holds since  $\ln$ is a concave function.
\end{proof}

% wiki and https://fractional-calculus.com/higher_derivative_composition.pdf
\begin{lemma}[Derivatives of composite function]\label{lem:deriv_fg}  Let $G(z)$ and $\phi(z)$ be two functions   satisfying $\phi(0)=0$. Let $m$ be an integer. Suppose the first $m$ derivatives $\{G^{(k)}(0)\}_{k=0}^m$ and $\{\phi^{(k)}(0)\}_{k=0}^m$ can be computed in time $t(n)$ where $n$ is a parameter. Then the first  $m$ derivatives of the composite function $G(\phi(z))$ at $z=0$, denoted as 
%$\left\{ \Big(g(\phi(0))\Big)^{(k)}\right\}_{k=1}^m$,
$$\left\{\frac{\partial^k }{\partial z^k}G(\phi (z))\Big|_{z=0}\right\}_{k=1}^m$$
can be computed in time $t(n)+O(m^4)$.
\end{lemma}

\begin{proof}
For integer $k$ and $r$, define the Bell polynomial to be
\begin{multline}
    B_{k,r}(\phi^{(1)}(z),\phi^{(2)}(z),\ldots ,\phi^{(k-r+1)}(z)) \\= \sum \frac{k!}{j_1!j_2!\ldots j_{k-r+1}!} \left(\frac{\phi^{(1)}(z)}{1!}\right)^{j_1}\left(\frac{\phi^{(2)}(z)}{2!}\right)^{j_2}\ldots \left(\frac{\phi^{(k)}(z)}{k!}\right)^{j_{k-r+1}} \label{eq:Bell}
\end{multline}
where the summation is 
\begin{align}
    j_1+j_2+\ldots +j_{k-r+1}=r \text{ and }  j_1+2j_2+\ldots +(k-r+1)j_{k-r+1}=k. \label{eq:cond_j}
\end{align} 
To compute the derivative of the composite function $G(\phi(z))$,
we will use the Faa di Bruno's formula which states that
\begin{align}
    \frac{\partial^k }{\partial z^k}G(\phi (z))\Big|_{z=0} &= \sum_{r=1}^k G^{(r)}(\phi(0)) \cdot B_{k,r}(\phi^{(1)}(0),\phi^{(2)}(0),\ldots ,\phi^{(k-r+1)}(0))\\
    &= \sum_{r=1}^k G^{(r)}(0) \cdot B_{k,r}(\phi^{(1)}(0),\phi^{(2)}(0),\ldots ,\phi^{(k-r+1)}(0)).\label{eq:faa}
\end{align}
  where the notation  $G^{(r)}(\phi(0))$ refers to $G^{(r)}(z)|_{z=\phi(0)}$ and
    the last equality comes from $\phi(0)=0$. 
For each $B_{k,r}$, we define the corresponding partial ordinary Bell polynomials as
\begin{equation}
    \hat{B}_{k,r}(y_{1}, \ldots, y_{k-r+1})=\frac{r!}{k!}B_{k,r}(x_{1},\ldots ,x_{k-r+1})
\end{equation}
where $y_i = \frac{x_i}{i!}$. They satisfy the recurrence formula
\begin{equation}
    \hat{B}_{k,r}(y_{1},\ldots ,y_{k-r+1})=\sum_{i=1}^{k-r+1}y_{i} \hat{B}_{k-i,r-1}(y_{1},\ldots ,y_{k-r+1-i}).
\end{equation}
Then after computing the first $m$ derivatives $\{G^{(k)}(0)\}_{k=0}^m$ and $\{\phi^{(k)}(0)\}_{k=0}^m$ in time $t(n)$, Algorithm \ref{alg:bellpoly} computes $B_{k,r}$ in time $O(k^2r)$. It suffices to compute $B_{m,r}$ for $r=1,\ldots ,m$ since all lower orders can be computed along the way, which takes total time $O(m^4)$. Therefore, the first $m$ derivatives of $G(\phi(z))$ at $z=0$ can be computed in time $t(n) + O(m^4)$.
\begin{algorithm}
\caption{Compute Bell polynomials \cite{taghavian2023fastalgorithmcomputingbell}}\label{alg:bellpoly}
\begin{algorithmic}[1]
\State Set $y_i \leftarrow x_i/i!$, $i=1,\ldots ,k-r+1$ 
\State Set $\hat{B}_{0,0}=1$, $\hat{B}_{i,0} = 0$, $i=1,\ldots ,k-r$
\For{$l=1,\ldots ,r$}
    \For{$i=l,\ldots ,k-r+l$}
        \State $\hat{B}_{i,l}(y_1,\ldots ,y_{k-r+1}) \leftarrow \sum_{j=1}^{i-l+1} y_j \hat{B}_{i-j,l-1}(y_1,\ldots ,y_{k-r+1})$
    \EndFor
\EndFor
\State Set $B_{k,r}(x_1,\ldots ,x_{k-r+1}) \leftarrow \frac{k!}{r!}\hat{B}_{k,r}(y_1,\ldots ,y_{k-r+1})$
\end{algorithmic}
\end{algorithm}
\end{proof}
 
\section{$\BPP$-$\hardness$ of additive-error approximation (Theorem \ref{thm:BPP})}	\label{appendix:BPPhard}
In this section, we prove the approximation problem in Theorem \ref{thm:BPP} is $\BPP$-$\hard$.
This proof is similar to Section 4.2 in \cite{arad2010quantum} and here we give a proof sketch. 

First we embed classical randomized computations into quantum circuits. Given a parameter $n$,
consider a quantum circuit of following form:
\begin{itemize}
    \item  Takes input as $\ket{0}^p\ket{+}^q$ for $p,q=poly(n)$.
    \item Applies a sequence of gates $Q=Q_L\ldots Q_1$, where $L=poly(n)$ and $\{Q_i\}_i$ are reversible gates on constant qubits.
    \item  Measure the first qubit in computational basis.
\end{itemize}
  Denote $p_0$ as the probability of getting measurement outcome $0$ in the first qubit. Suppose it is promised that either one of the following holds: 
  \begin{itemize}
      \item  Yes case: $p_0\geq 2/3$ 
      \item  No case: $p_0\leq 1/3$.
  \end{itemize}
One can check that the problem of given such a circuit, 
output Yes/No correctly with probability
greater than $2/3$ is $\BPP$-$\hard$.\footnote{Readers who are not familiar with randomized reduction may read Definition 7.19 in \cite{arora2009computational}.}  In other word 
\begin{claim}\label{claim:p0}
    An algorithm for estimating $p_0$ with high probability is $\BPP$-$\hard$.
\end{claim}

Similarly as \cite{arad2010quantum}, to write $p_0$ as a tensor network, we first define a related circuit $U$ on $p+q+1$ qubits: As shown in Figure \ref{fig:Q}, U firstly applies $Q$ to $|0^{\otimes(p+1)},+^{\otimes q}\rangle$, then copies the first qubit of $Q$ to the additional qubit by CNOT, and then applies $Q^{-1}$. 

\begin{figure}[H]
    \centering
\includegraphics[width=0.5\textwidth]{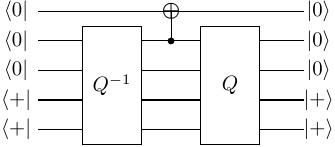}
    \caption{Circuit for $p_0$}
    \label{fig:Q}
\end{figure}

One can check that 
$$\langle 0^{\otimes(p+1)},+^{\otimes q} \big| U \big| 0^{\otimes(p+1)},+^{\otimes q}\rangle=p_0.$$
One can transform $\langle 0^{\otimes(p+1)},+^{\otimes q} \big| U \big| 0^{\otimes(p+1)},+^{\otimes q}\rangle$ to a tensor network  $T$ similarly as \cite{arad2010quantum},  then \begin{align}
    \chi(T) = \langle 0^{\otimes(p+1)},+^{\otimes q} \big| U \big| 0^{\otimes(p+1)},+^{\otimes q}\rangle=p_0,
\end{align} where 
\begin{itemize}
    \item Each reversible gate $Q_i$ on constant qubits is translated to a tensor $M^{[Q_i]}$, which is of constant rank (constant degree) and bond dimension $2$. Note that since $Q_i$ is a reversible gate, which is a permutation, thus we have $\|Q_i\|_1=1$. 
    \item We pair the input qubits on the left and right in Figure \ref{fig:Q}.  $|0\rangle\langle 0|$  is translated into a tensor  $M^{[0]}=\begin{bmatrix}
  1&0\\0&0
    \end{bmatrix}$, $|+\rangle\langle +|$  is translated into a tensor $M^{[+]}=\frac{1}{2}\begin{bmatrix}
   1&1\\1&1
    \end{bmatrix}$. Note that $\|M^{[0]}\|_1=\|M^{[+]}\|_1=1$.
\end{itemize}
Thus the approximation scale $\Delta_1$ in Eq.~(\ref{eq:scale}) is equal to $1$.
Thus for $\epsilon=1/poly(n)$,  the approximation problem in Theorem \ref{thm:BPP}, that is Eq.~(\ref{eq:BPP}), requires approximating  $\chi(T)=p_0$ within  precision $\epsilon$ with high probability, thus is $\BPP$-$\hard$ by Claim \ref{claim:p0}.

\bibliographystyle{alpha}
\bibliography{ref.bib}
\end{document}